\RequirePackage{fix-cm}
\documentclass[smallextended]{svjour3}
\smartqed

\usepackage{amsmath,amssymb,amsfonts}
\usepackage{graphicx}
\usepackage{booktabs}
\usepackage{multirow}
\usepackage{siunitx}
\usepackage{float}
\usepackage{placeins}
\usepackage{subcaption}
\usepackage{xcolor}
\usepackage{tikz}
\usetikzlibrary{arrows.meta,decorations.pathreplacing}
\usepackage{comment}
\usepackage[round,authoryear]{natbib}
\usepackage{url}
\usepackage{parskip}

\makeatletter
\let\c@lemma\c@theorem
\let\c@proposition\c@theorem
\let\c@corollary\c@theorem
\makeatother
\renewcommand\makeheadbox{}
\makeatother

\usepackage{xurl}
\usepackage[final, breaklinks=true]{hyperref}

\begin{document}

\title{Optimal Catastrophe Bond Design and its Applications to Climate Change Risk}
\subtitle{}

\titlerunning{Optimal Catastrophe Bond Design and its Applications} 

\author{Andres Medina \and Wei Wei}

\authorrunning{Medina \and Wei} 

\institute{Andres Medina \at
              Department of Mathematics, University of Illinois Urbana-Champaign, \\
              1409 W Green St, Urbana, IL 61801 \\
              \email{andres9@illinois.edu}           
           \and
           Wei Wei \at
              Program of Actuarial Science and Risk Management, University of Illinois Urbana-Champaign, \\
              Computing Applications Building, 605 E Springfield Ave, Champaign, IL 61820 \\
              \email{weiw@illinois.edu}
}

\date{}

\maketitle

\begin{abstract}
Catastrophe bonds have become a central instrument for transferring catastrophic risk to capital markets, yet their payout structures are typically specified exogenously rather than optimized. In this paper, we treat the indemnity function as the decision variable and study the optimal design of CAT bonds from the sponsor's perspective. We establish the optimality of a piecewise linear indemnity function under both expected-disutility and spectral-risk-measure frameworks and thereby reduce an infinite-dimensional design problem to a one-dimensional optimization over the attachment point. The optimal design also eliminates the potential moral hazard associated with the discontinuity of threshold payout structures. We characterize the optimal attachment point under six risk criteria and obtain explicit or semi-explicit solutions. We further derive bounds for the optimal attachment point and develop comparative statics that provide practical design insights. Using U.S. severe storm losses from NOAA's Billion-Dollar Weather and Climate Disasters database for calibration, we show that the optimal CAT bond designs achieve substantial risk reductions across most risk criteria. 






\end{abstract}

\section{Introduction}
\label{intro}

A Catastrophe Bond (CAT bond) is a financial instrument that transfers the risk of extreme events from sponsoring institutions, such as governments and insurance companies, to capital market investors. By securitizing catastrophe risk, sponsors obtain additional loss-absorbing capital and reduce reliance on traditional reinsurance, whose cost and availability often deteriorate after major disasters as markets harden and capacity contracts \citep{cummins2009convergence,froot2001market}. Unlike conventional reinsurance, CAT bonds are fully collateralized, and thus eliminate counterparty risk and ensure that funds remain available even during periods of market stress. They also typically provide multi-year coverage, which enables sponsors to lock in pricing and protection over extended time periods \citep{polacek2018catastrophe}. 
\\

 The CAT bond market emerged in response to shortcomings in traditional catastrophe risk financing exposed by Hurricane Andrew in 1992, whose unprecedented losses contributed to numerous insurer failures and revealed capacity constraints in reinsurance markets \citep{NOAA2024,mcchristian2012hurricane}. The search for alternative risk-transfer mechanisms subsequently led to several forms of catastrophe-linked securities and, ultimately, to the development of the CAT bond market \citep{cummins2012catbonds}. Since then, CAT bonds have evolved into a major segment of insurance-linked securities, with risk capital outstanding reaching \$65.5 billion as of mid-2026. Their growth reflects demand from sponsors seeking additional catastrophe-risk capacity and investors attracted by the asset class's diversification benefits and risk-return profile \citep{haffar2022dependence,drobetz2020catbonds,artemis2026}.
\\

 In the early development of the catastrophe bond market, contract design was dominated by threshold payout structures, under which bond principal is immediately exhausted once a predefined trigger hits a specified threshold. Such contracts arise naturally as a standard feature of CAT Bond transactions \citep{cummins2008catbonds, cummins2012catbonds},  because they are intuitive, easy to communicate to investors, and relatively transparent in implementation. Their principal drawback, however, stems from the discontinuity at the trigger point, where a small change in the underlying loss measure can generate a disproportionately large change in contractual payouts. Such a cliff effect can encourage costly disputes over trigger determination and introduce additional frictional costs into the contracting process. Similar concerns have been extensively documented in both insurance and finance, see for example, \citet{bernard2010insurance} and \citet{kisgen2006credit}. In the CAT Bond context, our analysis provides a strong theoretical justification against  threshold structures. Specifically, we show that over single-layer indemnities with an unrestricted slope, maximizing the sponsor's expected net loss drives the design to a threshold contract (see Appendix~\ref{app:max_expected_net_loss}). This result suggests that the simplicity of threshold contracts may come at a substantial efficiency and risk reduction cost.

 Motivated by the limitations of threshold structures, market participants have increasingly adopted smoother payout mechanisms. Among the most prominent alternatives are piecewise-linear indemnity functions, which preserve much of the transparency and implementability of threshold contracts while mitigating the abrupt payout jumps associated with cliff effects. Several recent CAT Bond transactions illustrate this trend. For example, recent World Bank-sponsored issuances for Mexico (US \$595 million) and Chile (US \$350 million) employ stepped or sliding-scale payout mechanisms which link principal reductions more continuously to event severity and thereby reduce the abrupt payout discontinuities associated with threshold structures \citep{worldbank2025mexico, worldbank2023chile_press}. These examples suggest that piecewise-linear structures have emerged in practice as a natural compromise between contractual simplicity and efficient risk transfer.

 The growing prevalence of such contracts has also attracted academic attention. For instance, \citet{baryshnikov1998pricing} price CAT bonds with threshold payouts, and \citet{nowak2013pricing} extend valuation to piecewise-linear payoff structures, in both cases, treating the indemnity function as an exogenously specified contractual feature. However, an important question remains unanswered: while piecewise-linear contracts appear attractive in practice and tractable in valuation models, is their widespread adoption supported by theoretical justification, or are there superior contract forms that remain unexplored?

In order to answer this question, we formulate CAT Bond design as a functional optimization problem in which the indemnity function itself serves as the decision variable. Rather than restricting attention to a particular family of payout structures, we optimize over a broad admissible class of indemnity functions and establish that the optimal indemnity function must take a piecewise-linear form  under a wide range of economically meaningful objective criteria. This result suggests that the growing use of piecewise-linear CAT bond structures is supported not only by practical considerations, but also by economic optimality under a broad class of risk assessment criteria.

The CAT Bond design problem considered in this paper is closely related to the classical problem of optimal insurance and reinsurance design. In both settings, the indemnity function serves as the control variable, and the objective is to determine the most effective risk-transfer arrangement under a given performance criterion. Consequently, the two problems share many structural similarities. In particular, the optimal CAT Bond indemnity function derived in this paper takes a piecewise-linear form under which coverage becomes effective only after losses exceed a certain threshold, referred to as the attachment point. This feature resembles the stop-loss contract, which is known to be optimal in a wide variety of insurance and reinsurance settings. The study of optimal insurance and reinsurance contracts has a long history, dating back to the seminal works of \citet{Borch1960} and \citet{Arrow1963}. Subsequent research has considerably broadened the scope of the field, covering expected-utility criteria, risk-measure optimization, distortion-based models, and Pareto-optimal risk-sharing problems; see \citet{CaiChi2020} for a comprehensive review.

Despite these similarities, the methodologies developed for optimal insurance and reinsurance design cannot be carried over directly to the CAT Bond setting. The key distinction is that CAT Bond coverage is inherently limited by the issuance volume. Whereas a conventional stop-loss contract provides protection against all losses above a deductible, a CAT bond covers only a finite layer of losses. As a result, the central design question is no longer merely how much risk should be transferred, but rather which portion of the loss should receive protection.  Determining the optimal location of this coverage layer, and understanding how it responds to changes in pricing, risk preferences, and loss characteristics, requires analytical tools that differ substantially from those used in the classical optimal reinsurance literature.

 Building on the piece-wise linear structural characterization, we further derive the optimal indemnity strategy by identifying the optimal attachment point.  The results provide quantitative guidance for issuers to reduce catastrophe-risk exposure while maintaining cost-effective protection. We then conduct a comparative-statics analysis to examine how the optimal CAT Bond design responds to changes in model parameters. The analysis clarifies how the economic forces shape the optimal contract and provides guidance for constructing robust designs across a range of economic environments rather than for a single calibrated scenario.

 Finally, we implement the proposed framework using U.S. severe-storm loss data. The empirical study illustrates how the theoretically optimal designs can effectively reduce risk exposure and demonstrates the practical relevance of our approach. Beyond validating the theoretical findings, the case study also identifies several questions that remain outside the scope of the current framework and points to promising directions for future research.

The remainder of the paper is organized as follows: Section~\ref{sec:model-setup} introduces the model, formulates the design problem, and establishes its well-posedness. Section~\ref{sec:optimal_indemnity} characterizes the optimal indemnity form. Section~\ref{sec-optimal_attachment_point} studies the optimal attachment-point problem in detail. Section~\ref{sec:case_study} presents the case study and numerical implementation. Section~\ref{sec:conclussions} concludes and discusses directions for future research.

\section{Model formulation}
\label{sec:model-setup}
Suppose the aggregate loss of an insurance portfolio is modeled by a nonnegative continuous random variable $X$ with a cumulative distribution function $F(x)$. Assume $X$ has a finite expected value and a decreasing probability  density function. This assumption of decreasing pdf is a standard practice in actuarial science \citep[see, e.g.,][]{embrechts1997modelling, klugman2012loss} to capture the diminishing probabilities of extreme losses. The insurer, acting as the sponsor, seeks protection from the capital market by issuing catastrophe bonds (CAT bonds). 

We use the direct loss as the trigger mechanism. For a unit of CAT bond (with face amount $\$1$), let $I(x)$ denote the indemnity ratio paid to the insurer given a loss realization $x$. By definition, $0 \le I(x) \le 1$ and $I(x)$ is a non-decreasing function. Under this indemnity schedule, the investor's potential loss per unit CAT bond is $I(X)$, resulting in a final payoff of $1-I(X)$. For analytical simplicity, we assume there is no coupon and the risk-free interest rate is zero, while positive interest rate would not add technical complexity in the analysis.

To price the CAT bond, we adopt the expected value premium principle. Specifically, the indemnity provided by investors to the insurer can be viewed as a reinsurance contract, where the coverage per unit, $I(X)$, is priced as $(1+\tau_2)\mathbb{E}[I(X)]$. This term represents the risk premium required by the investors. Subtracting this premium from the principal amount yields the price of a unit CAT bond:
\begin{equation*}
       P_2= 1- (1+\tau_2)\mathbb{E}[ I(X)].
\end{equation*}

 Let $P_1$ denote the total premium collected from the insureds for covering the risk $X$. We continue to use the expected value premium principle and calculate $P_1$ by $(1+\tau_1)\mathbb{E}[X]$.  To reflect the fact that issuing a CAT bond as a reinsurance instrument should be more costly than direct insurance, we impose the following constraints for the two safety loading coefficients: $\tau_2\ge \tau_1 \ge 0$. In the following analysis, the coefficients $\frac{\tau_2}{1+\tau_2}$ and $\frac{1}{1+\tau_2}$ appear frequently; for convenience, we denote 
 \begin{equation*}
     \rho_0 = \frac{\tau_2}{1+\tau_2} \quad \text{and}\quad 1-\rho_0 = \frac{1}{1+\tau_2}. 
 \end{equation*} 
 By viewing the investor as a reinsurer,  $1-\rho_0$ and $\rho_0$ can be respectively interpreted as the expected loss ratio and the expected expense ratio associated with providing reinsurance coverage through the purchase of the CAT bond.

 Assume the insurer issues $w$ units of CAT bonds. At the maturity of the CAT bonds,  the net loss  of the insurer  under the indemnity scheme $I$ is 
\begin{eqnarray*}
\label{eqn:net_loss}
      L_I &=& X - P_1 + w(1-I(X)) - wP_2    \nonumber \\  
          &=& X - (1+\tau_1)\mathbb{E}[X]  + w(1+\tau_2)\mathbb{E}[I(X)]- w I(X) \label{netloss} 
\end{eqnarray*}

The goal is to find the optimal indemnity scheme so as to minimize the net loss in some sense. Generally, 
the optimization problem can be formulated as follows
\begin{equation*}
      \inf_{ I \in \mathfrak{C}} \mathcal{R}({L_I}),
\end{equation*} 
 where $\mathfrak{C}$ denotes the set of feasible indemnity schemes and $\mathcal{R}$ represents a functional mapping a random variable to a real number, which is used to measure the riskiness of the net loss random variable. 

Since $I(x)$ represents the indemnity ratio, it should be bounded between $0$ and $1$.Intuitively, the indemnity provided by CAT bond investors should increase with the realized loss severity. Moreover, to mitigate moral hazard, the marginal indemnity should not exceed the marginal increase in the underlying loss, which implies $I^\prime(x) \le 1/w$.  Motivated by these practical considerations, we define the following class of admissible strategies:

\begin{equation}\label{control_set}
\mathfrak{C}_0 = \left\{ I:[0,\infty)\to[0,1] \;\middle|\; 
0 \le I^\prime(x) \le \frac{1}{w}
\right\}
\end{equation}

The choice of $\mathcal{R}$ can be very general. In this paper, we study the optimization problem in two specific frameworks, by setting $\mathcal{R}$ to be an expected disutility and a spectral risk measure respectively. 

The expected disutility framework provides a natural extension of expected utility theory for analyzing decisions involving losses. Rather than modeling preferences over final wealth, this approach directly evaluates random losses $L$ through an increasing and convex loss disutility function $v(\cdot)$, with the objective of minimizing $\mathbb{E}[v(L)]$. Early conceptual foundations can be traced to  \cite{theil1976disutility}.

Under the expected disutility framework the optimization problem is formulated as follows: 
\begin{equation}\label{P-EDU}
    \inf_{I \in \mathfrak{C}_0} J_v(I) = \inf_{I \in \mathfrak{C}_0} \mathbb{E}[v(L_I)]
\end{equation}
    for an increasing convex and continuously differentiable disutility function $v:\mathbb{R}\to \mathbb{R}$. 
 Note that the increasingness and continuous differentiability of $v$ are imposed only to streamline the presentation of the main results. Although our primary focus is on this class of disutility functions, we also consider several functions of practical relevance that do not strictly satisfy these conditions.
 Whenever our key results extend to these alternative functions, we will explicitly indicate so. Throughout this paper, unless otherwise stated, a disutility function is assumed to be increasing, convex, and continuously differentiable.

In addition to the expected disutility framework, we consider an alternative approach based on spectral risk measures. This class of risk measures has gained considerable attention in actuarial science for its flexibility in capturing risk preferences. Foundational developments are due to \cite{Acerbi2002}, with further applications in insurance pricing and portfolio optimization by \cite{Dowd2006}. 

\begin{definition}
\label{def:spectralriskmeasure}
    Let $Y$ be a random variable and $\operatorname{VaR}_p(Y)$ denote its value at risk at level $100p\%$. A spectral risk measure of $Y$ is defined as 
    \begin{equation*}
        M_\phi(Y) = \int_{0}^1 \phi(p) \operatorname{VaR}_{p}(Y) d p,
    \end{equation*}
    where $\phi: [0,1]\to [0,\infty)$, referred to as risk spectrum, is an increasing function with $\int_0^1 \phi(p)dp=1$.  
\end{definition}

The risk spectrum function $\phi(p)$ specifies the weight a decision-maker assigns to each quantile level when evaluating a risk. For example, if $\phi(p)$ is identically equal to $1$, the decision-maker treats all quantiles equally, evaluating the risk by its expected value. Alternatively, if $\phi(p) = \frac{\mathbb{I}\{p>\beta\}}{1-\beta}$, the decision-maker focuses exclusively on quantiles above the threshold level of $\beta$, weighting them equally, which  leads to the Tail Value-at-Risk (TVaR) measure.

Spectral risk measures are coherent according to \citet{artzner1999coherent} and thus possess the properties of  monotonicity, translation invariance, positive homogeneity, and subadditivity.  Notably, spectral risk measures are equivalent to the class of distortion risk measures introduced by \citet{wang1996premium} when the underlying distortion function is concave, see \citet{gzyl2006relationship}. In this paper, we adopt the spectral representation for its tractability in optimization.

With this definition in place, we proceed to formulate the optimization problem under the spectral risk measure framework.

\begin{equation}\label{P-SPEC}
    \inf_{I \in \mathfrak{C}_0} M_\phi(L_I)
    \end{equation}

To ensure these problems are well-posed and meaningful, we must verify that the infimum does not diverge to negative infinity. The following two propositions demonstrate that, under both the expected disutility and the spectral risk measure frameworks, the objective functions are uniformly bounded from below by a finite constant.

\begin{proposition}
For any convex disutility function $v$, there exists a constant $B>-\infty$ such that $J_v(I) \ge B$ for all $I \in \mathfrak{C}_0$. 
\end{proposition}

\begin{proof}
    Since $0\le I(x) \le 1$, we have 
\[   \mathbb{E}[L_I]
= -\tau_1 \mathbb{E}[X] + w \tau_2 \mathbb{E}[I(X)] \in [-\tau_1 \mathbb{E}[X],\,-\tau_1 \mathbb{E}[X] + w \tau_2] := [a, b].
\]
According to Jensen's inequality, it holds that 
\[
J_v(I) = \mathbb{E}[v(L_I)] \ge v(\mathbb{E}[L_I]) \ge \inf_{y \in [a,b]} v(y). 
\]
Recall that $v$ is convex and therefore continuous, and $[a,b]$ is a closed interval. According to Weierstrass's theorem, $v$ obtains its minimum at some $x_0\in [a, b]$ and hence $\inf_{y \in [a,b]} v(y) = v(x_0)$ is finite. 
\end{proof}

\begin{proposition}
For any risk spectrum $\phi$ such that $M_\phi(X)>-\infty$, there exists a constant $B>-\infty$ such that $M_\phi(L_I)\ge B$ for all   $I\in\mathfrak C_0$. 
\end{proposition}

\begin{proof}
Noting that $0\le I(X)\le 1$ and $w\ge0$, we have
\[
L_I
=
X-(1+\tau_1)\mathbb{E}[X]
+(1+\tau_2)\mathbb{E}[wI(X)]
-wI(X)
\ge
X-(1+\tau_1)\mathbb{E}[X]-w.
\]
By monotonicity and translation invariance of $M_\phi$,
\[
M_\phi(L_I)
\ge
M_\phi\!\left(X-(1+\tau_1)\mathbb{E}[X]-w\right)
=
M_\phi(X)-(1+\tau_1)\mathbb{E}[X]-w>-\infty. 
\]
\end{proof}

\section{Optimal indemnity form}
\label{sec:optimal_indemnity}

\begin{theorem}\label{thm:optform}
The optimal indemnity function \(I\) for Problems \eqref{P-EDU} and \eqref{P-SPEC}
has the form
\begin{equation}\label{optI}
I_d(x)=\frac{(x-d)_+}{w}\wedge 1 .
\end{equation}
\end{theorem}

\begin{proof}
To prove the optimality of $I_d$  in solving Problem \eqref{P-EDU}, i.e., minimizing $\min \mathbb{E}[v(L_I)]$ for any convex function $v$, it suffices to show that,  for any \(I\in\mathfrak C_0\), there exists \(d\in\mathbb R\) such that \(L_d\le_{cx}L_I\), where \(L_d\) denotes the net loss under \(I_d\).  Since spectral risk measures preserve the convex order  \citep{dhaene2000comonotonicity}, the same argument proves the optimality of  $I_d$  in solving Problem \eqref{P-SPEC}.


For each $x$, $I_d(x)$ is continuous in $d$, with $\lim_{d\to-\infty} I_d(x)=1$ and $\lim_{d\to\infty} I_d(x)=0$. Since $0\le I_d\le 1$, we have $\lim_{d\to-\infty} \mathbb{E}[I_d(X)]=1$ and $\lim_{d\to\infty} \mathbb{E}[I_d(X)]=0$ according to  
the dominated convergence theorem. 
Since $\mathbb{E}[I(X)]\in[0,1]$, intermediate value theorem yields that there exists  $d\in\mathbb{R}$ such that
\begin{equation}\label{eq:meanmatch}
\mathbb{E}[I_d(X)]=\mathbb{E}[I(X)], \quad \text{and thus} \quad \mathbb{E}[L_d]=\mathbb{E}[L_I].
\end{equation}
Define
\[
\Delta(x):=L_I(x)-L_d(x) = 
\begin{cases}
-wI(x)\le 0, & 0\le x<d,\\[4pt]
x-d-wI(x), & d\le x<d+w,\\[4pt]
w(1-I(x))\ge 0, & x\ge d+w.
\end{cases}
\]
Note that \(\Delta'(x)=1-wI'(x)\ge 0\) for any $x \in (d, d+w)$ due to the marginal indemnity constraint of $I$ as assumed in  \eqref{control_set}. Hence, \(\Delta\) is non-positive on \([0,d)\), is increasing on \([d,d+w)\), and is nonnegative on \([d+w,\infty)\). Therefore, \(\Delta\) changes sign at most once, from negative to positive, so \(L_d(x)\) and \(L_I(x)\) cross at most once. Since both \(L_d(X)\) and \(L_I(X)\) are increasing functions of $X$, this single-crossing property carries over to their distribution functions, meaning that there exists \(\ell_0\) such that
\begin{equation}\label{eq:single_crossing}
F_{L_d}(\ell)\le F_{L_I}(\ell) \text{ for } \ell<\ell_0, \qquad F_{L_d}(\ell)\ge F_{L_I}(\ell) \text{ for } \ell>\ell_0.
\end{equation}
Following \eqref{eq:meanmatch} and \eqref{eq:single_crossing}, the hypotheses of Ohlin's lemma \citep{Ohlin1969} hold, which in turn yields  \(L_d\le_{cx}L_I\).
\end{proof}



With the optimal form of indemnity identified by Theorem \ref{thm:optform}, the loss random variable depends solely on $d$ and is thus denoted as $L_d$. Specifically, 
\begin{equation*}
    L_d = X- P_1 + (1+\tau_2)\mathbb{E}[(X-d)_+\wedge w] - (X-d)_+\wedge w.
\end{equation*}
Denote 
\begin{equation*}
m(d) = d- P_1 + (1+\tau_2)\mathbb{E}[(X-d)_+\wedge w]. 
\end{equation*}
$L_d$ can be rewritten as 
\begin{equation}\label{piecewiseld}
    L_d = X-d - (X-d)_+\wedge w + m(d).
    \end{equation}
Clearly, $L_d$ is an increasing\footnote{Throughout this paper, "increasing" means non-decreasing and "decreasing" means non-increasing.} and continuous function of $X$.  To emphasize this, we also write $L_d$ as $L_d(X)$, with $L_d(x)$ specified by 
 \begin{equation*}
    L_d (x) = x-d - (x-d)_+\wedge w + m(d).      
   \end{equation*}
 Note that   
 \begin{equation}\label{ld'}
     \frac{\partial}{\partial d}L_d(x) = m'(d) - 1 + \mathbb{I}\{d < x \le d+w\}.
 \end{equation}

Figure \ref{fig:net_loss_piecewise} shows the graph of $L_d$ as a function of $X$.  When the CAT bond is triggered but not yet exhausted, that is, when the initial loss $X$ falls into the interval $[d, d+w)$, the marginal indemnity received by the insurer exactly offsets the marginal increase in the loss, thereby keeping the net loss constant at the plateau level $m(d)$. Outside this interval, the CAT bond is either not triggered (for $X < d$) or exhausted (for $X \ge d+w$); in either case, the total indemnity remains constant at $0$ or $w$, respectively, so the net loss becomes a linear function of $X$ with slope one on both pieces. Conventionally, $d$ and $d+w$ are referred to as the attachment point and the exhaustion point, and $[d, d+w]$ as the coverage layer.  

\begin{figure}[htbp]
\centering
\begin{tikzpicture}[scale=1.05, >=stealth]

    \pgfmathsetmacro{\d}{3.2}
    \pgfmathsetmacro{\w}{2.4}
    \pgfmathsetmacro{\m}{2.2}
    \pgfmathsetmacro{\xmax}{8}

    \pgfmathsetmacro{\yleft}{\m-\d}
    \pgfmathsetmacro{\xright}{\d+\w}
    \pgfmathsetmacro{\yright}{\m + (\xmax-\d-\w)}

    \draw[thick,->] (0,0) -- (8.3,0) node[below] {$X$};
    \draw[thick,->] (0,0) -- (0,5.2) node[left] {$L_d(X)$};

    \draw[very thick]
        (0,\yleft) -- (\d,\m) -- (\xright,\m) -- (\xmax,\yright);

    \draw[dashed] (\d,0) -- (\d,\m);
    \draw[dashed] (\xright,0) -- (\xright,\m);

    \draw (\d,0.08) -- (\d,-0.08) node[below=4pt] {$d$};
    \draw (\xright,0.08) -- (\xright,-0.08) node[below=4pt] {$d+w$};

    \draw (0.08,\m) -- (-0.08,\m) node[left=4pt] {$m(d)$};

    \fill (\d,\m) circle (1.5pt);
    \fill (\xright,\m) circle (1.5pt);

\end{tikzpicture}
\caption{Net loss $L_d(X)$ as a function of the underlying loss $X$.}
\label{fig:net_loss_piecewise}
\end{figure}
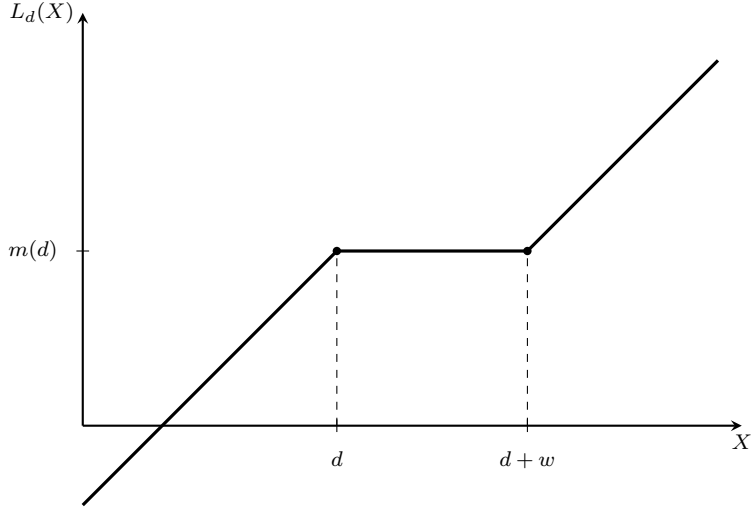

The plateau level $m(d)$ is of independent interest. It represents the level at which the insurer's "anticipated" level of net loss, when the CAT bond functions as intended. For a sound CAT bond design, the plateau level should be nonnegative; a negative plateau would imply that the CAT bond generates a profit for the insurer, which may be viewed unfavorably from a risk management or regulatory standpoint. We shall investigate the sign of the plateau level in the following analysis when appropriate.

The behavior of $m(d)$ plays an essential role in the derivation of the optimal solutions. We summarize its key properties below.
\begin{lemma}\label{l:md}The function $m(d)$ satisfies the following properties:
\begin{itemize}
\item[(i)] Its first order derivative is given by
$$ \frac{\partial }{\partial d} m(d) = 1 - (1+\tau_2)\mathbb{P}\{d < X \le d+w\} $$
\item[(ii)] $m(d)$ is convex on the interval $[0, \infty)$.
\end{itemize}
\end{lemma}
\begin{proof}
    (i).    Note that $\frac{\partial^-}{\partial d}((X-d)_+\wedge w) = -\mathbb{I}\{d\leq X <  d+w\}$. Following the dominated convergence theorem, we have
    \begin{eqnarray*}
        \frac{\partial ^-}{\partial d} m(d)
        &=& 1 + (1+\tau_2)\mathbb{E}\left[ \frac{\partial^-}{\partial d} (X-d)_+\wedge w \right] \\
        &=& 1 + (1+\tau_2)\mathbb{E}[-\mathbb{I}\{d\le X <  d+w\}] \\
        &=& 1 - (1+\tau_2)\mathbb{P}\{d \le X <  d+w\}.    
    \end{eqnarray*}
    Similarly, $\frac{\partial^+}{\partial d} m(d) = 1-  (1+\tau_2)\mathbb{P}\{d <  X \le d+w\}$. Since $X$ is assumed to have a continuous distribution, we have $\frac{\partial^+}{\partial d} m(d) = \frac{\partial^-}{\partial d} m(d)$. Thus $m(d)$ is differentiable everywhere and its derivative takes the expression of $\frac{\partial^-}{\partial d} m(d)$. 

(ii) immediately follows from (i) by noting that $X$ has a decreasing pdf and thus $\mathbb{P}\{d < X \le d+w\}$ is decreasing in $d \ge 0$.
    
 \end{proof}

\begin{remark}\label{r:derivative} 
As demonstrated in the above proof, differentiating quantities associated with capped loss  frequently requires evaluating left- and right-hand derivatives. However, due to the continuity assumption on $X$, this distinction does not impact the final result. To streamline the presentation, we will hereafter omit explicit references to one-sided derivatives and simply use standard derivative notation throughout the remainder of the paper. 
\end{remark}

Our primary objective is to determine the optimal attachment point under various models. To this end, it is instructive to develop an intuition for how variations in the attachment point impact the net loss random variable, $L_d$. Specifically, we examine the behavior of $L_d$ as the attachment point $d$ decreases. Lowering $d$ triggers coverage earlier, producing two countervailing effects: first, it reduces the anticipated level of net loss; second, it expands coverage, thereby increasing the associated cost for issuing CAT bonds. In the subsequent analysis of different models, we demonstrate how these competing forces interact, ultimately yielding an optimal outcome when they are balanced. 



With the optimal indemnity form specified by \eqref{optI}, Problems \eqref{P-EDU} and \eqref{P-SPEC} reduce to finding the optimal attachment point $d$. Although $d$ can theoretically take any real value, a value below $-w$ would result in $I_d(x)$ being identically equal to 1, rendering the problem trivial. In light of this, we restrict the domain to $d \ge -w$ in the formulation of the problem.

Note that for a negative attachment point, $I_d(0) > 0$, implying that the sponsor would receive compensation even when no loss is incurred. This is economically counterintuitive. Consequently, it is common practice in the literature to impose the constraint $I(0)=0$ or $d \ge 0$ for admissibility. We remark that, although we do not explicitly impose this constraint a priori, we will later show that under our framework, the optimal attachment point is inherently nonnegative, thereby satisfying this condition. In this sense, our modeling framework naturally accommodates practical application requirements.

Formally, Problems \eqref{P-EDU} and \eqref{P-SPEC} reduce to 
\begin{eqnarray}
\label{P-EDU'}
    \inf_{d \ge -w} J_v(d) &:=& \mathbb{E}[v(L_d)]\\
    \inf_{d \ge -w} J_\phi(d) &:=& \ M_\phi(L_d) 
    \label{P-SPEC'}
\end{eqnarray}
We shall study these two problems respectively in the following section. 

\section{The optimal attachment point}
\label{sec-optimal_attachment_point}

\subsection{Minimizing expected disutility} \label{subsec-edu}

In this subsection, we study Problem \eqref{P-EDU'} to find the optimal attachment point to minimize the expected disutility. To this end,  we impose the following technical conditions on $v$:
\begin{eqnarray}
    \label{v-cond}
 &&   \mathbb{E}[\max\{|v(X-P_1-w)|,|v(X-P_1+(1+\tau_2)w)|\}]<\infty \\
 &&   \mathbb{E}[\max\{|v^\prime(X-P_1-w)|,|v^\prime(X-P_1+(1+\tau_2)w)|\}]<\infty     \label{v'-cond}
\end{eqnarray}
 Noting that 
  $X-P_1 -w \le L_d \le X-P_1 + (1+\tau_2)w$, Condition \eqref{v-cond} guarantees that the expectation $\mathbb{E}[v(L_d)]$ always exists, and Condition \eqref{v'-cond} validates of the application of the dominated convergence theorem in the analysis of the derivative of the objective function $J_v(d)$. 

To determine the optimal attachment point, it is necessary to calculate the derivative of the objective function with respect to $d$. For this purpose, we summarize the expression of derivatives of some key quantities below.

\begin{lemma}\label{l:derivatives}
The derivatives of $J_v(d)$ with respect to $d$ is
    \begin{equation}
        \frac{\partial }{\partial d}\, J_v(d) = \mathbb{P}\{d<X\le d+w\} \times h_v(d)\label{Jv'}
    \end{equation}
    where
    \begin{eqnarray}\label{hfun}
    h_v(d) &=&  \mathbb{E}[v^\prime(L_d)|d<X\le d+w] - (1+\tau_2)\mathbb{E}[v^\prime(L_d)]  \\
    &=& v^\prime(m(d))  - (1+\tau_2)\mathbb{E}[v^\prime(L_d)] \nonumber
    \end{eqnarray}
    is defined on $d\in [-w,\infty)$.
\end{lemma}
\begin{proof}
    Following a similar argument as in the proof of Lemma \ref{l:md}, we have    
    \begin{eqnarray*}
        \frac{\partial }{\partial d} J_v(d) &=& \frac{\partial }{\partial d} \mathbb{E}[v(L_d)] 
        = \mathbb{E}\left[v^\prime(L_d)\times \frac{\partial }{\partial d} L_d\right]\\
        &=& \mathbb{E}[v^\prime(L_d)(\mathbb{I}\{d<X\le d+w\} - (1+\tau_2)\mathbb{P}\{d<X\le d+w\})]\\
        &=&  \mathbb{P}\{d<X\le d+w\}(\mathbb{E}[v^\prime(L_d)|d<X\le d+w] - (1+\tau_2)\mathbb{E}[v^\prime(L_d)])\\
       &=&  \mathbb{P}\{d<X\le d+w\}\big(v^\prime(m(d)) - (1+\tau_2)\mathbb{E}[v^\prime(L_d)]\big).
    \end{eqnarray*}
\end{proof}

\begin{lemma}\label{l:hprop}
    The function $h_v(d)$ as defined in \eqref{hfun} has the following properties:
    \begin{itemize}
        \item[(i)] $h_v(d)$ is increasing in $d \in [-w, \infty)$.
        \item[(ii)]  $h_v(d) \le 0$ for any $d\in[-w, \operatorname{VaR}_{\rho_0}(X)]$, where $\rho_0=\frac{\tau_2}{1+\tau_2}$.
        \item[(iii)] $\lim_{d\to \infty} h_v(d)= v^\prime(\infty) - (1+\tau_2)\mathbb{E}[v^\prime(X-P_1)]$.
    \end{itemize}
\end{lemma}

\begin{proof}
(i)  Since $v^\prime$ is continuous and increasing, $v^\prime$ is differentiable almost everywhere. Recalling that $h_v(d)=v^\prime(m(d)) - (1+\tau_2)\mathbb{E}[v^\prime(L_d)]$, we have 
\begin{eqnarray*}
    h_v^\prime(d) &=& v^{\prime\prime}(m(d))\,m^\prime(d)  - (1+\tau_2)\mathbb{E}[v^{\prime\prime}(L_d)\times \frac{\partial}{\partial d} L_d] \nonumber\\
    &=& v^{\prime\prime}(m(d))(1-(1+\tau_2)\mathbb{P}\{d<X\le d+w\}) \nonumber\\
    && -\, (1+\tau_2) \mathbb{E}[v^{\prime\prime} (L_d)\times (\mathbb{I}\{d < X \le d+w\} - (1+\tau_2)\mathbb{P}\{d<X \le d+w\})] \nonumber\\
    &=& v^{\prime\prime}(m(d)) \left(1-2(1+\tau_2)\mathbb{P}\{d<X\le d+w\} + (1+\tau_2)^2 \mathbb{P}\{d<X\le d+w\} \frac{\mathbb{E}[v^{\prime\prime}(L_d)]}{v^{\prime\prime}(m(d))}\right) \\
    &\ge& v^{\prime\prime}(m(d)) \left(1-2(1+\tau_2)\mathbb{P}\{d<X\le d+w\} + (1+\tau_2)^2 (\mathbb{P}\{d<X\le d+w\})^2\right) \label{h'ineq}\\
    &=&  v^{\prime\prime}(m(d)) \left(1-(1+\tau_2)\mathbb{P}\{d<X\le d+w\}\right)^2 \ge 0.
\end{eqnarray*}
where the first inequality  follows from  
\begin{equation*}
   \mathbb{E}[v^{\prime\prime}(L_d)] \ge     \mathbb{E}[v^{\prime\prime}(L_d)\times \mathbb{I}\{d<X\le d+w\}] = v^{\prime\prime}(m(d))\times \mathbb{P}\{d<X\le d+w\}
\end{equation*}
due to the convexity of $v$.


(ii)
 Noting that $L_d \ge m(d)$ for any $X\ge d$ and $v^\prime(\cdot) \ge 0$, we have 
     \begin{eqnarray*}
        h_v(d) &=&  v^\prime(m(d))  - (1+\tau_2)\mathbb{E}[v^\prime(L_d)] \\
             &\le&  v^\prime(m(d))  - (1+\tau_2)\mathbb{E}[v^\prime(L_d) \mathbb{I}\{X>d\}]\\
             &\le&  v^\prime(m(d))  - (1+\tau_2)\mathbb{E}[v^\prime(m(d)) \mathbb{I}\{X>d\}]\\
             &=&  v^\prime(m(d))(1  - (1+\tau_2) \mathbb{P}\{X>d\}).
     \end{eqnarray*}
For any $d \le \operatorname{VaR}_{\rho_0}(X)$, we have $\mathbb{P}\{X \le d\} \le \rho_0$ and thus $\mathbb{P}\{X>d\}  \ge 1 - \rho_0 = \frac{1}{1+\tau_2}$, which implies $h_v(d)\le 0$.

(iii) immediately follows from (ii) and the monotone convergence theorem. 
\end{proof}

\begin{remark}\label{r1-d*>0}
 For a convex but not necessarily increasing function $v$, it can be proved similarly as (ii) that $h_v(d)\le 0$ for any $d\le 0$.
\end{remark}

\begin{theorem}\label{t:opt-d}
    Consider Problem \eqref{P-EDU'} with a convex function $v$ satisfying \eqref{v-cond} and \eqref{v'-cond}.  The optimal attachment point $d^*$ is given by 
    \begin{equation}\label{opt-d}
    d^* = \inf\{d: h_v(d) > 0\},    
    \end{equation}
with the convention of $\inf\emptyset = \infty$.
         
\end{theorem}

\begin{proof}
Since $h_v(d)$ is increasing on $d\ge -w$ according to Lemma \ref{l:hprop}, it holds that $h_v(d)\le 0$ for any $d<d^*$ and $h_v(d)> 0$ for any $d>d^*$, and, by \eqref{Jv'}, the same statement holds for $J_v^\prime(d)$.  That means  $J_v(d)$ is decreasing on $[-w, d^*)$ and strictly increasing on $(d^*, \infty)$, which implies $J_v(d)$ obtains its minimum at $d=d^*$. 

\end{proof}

As discussed in Section \ref{sec:optimal_indemnity}, a decrease in the attachment point $d$ impacts the net loss random variable $L_d$ in two ways: lowering the anticipated net loss $m(d)$ while simultaneously raising the associated cost. In the expected disutility framework, the former effect reduces the expected disutility, yielding a marginal benefit of $v^\prime(m(d))$; conversely, the latter effect increases the expected disutility, imposing a marginal cost of $(1+\tau_2)\mathbb{E}[v^\prime(L_d)]$. Starting from $d=\infty$ (representing a no-issuance strategy), the marginal benefit exceeds the marginal cost. To minimize expected disutility, we must continually lower the attachment point until this is no longer the case, ultimately leading to the optimal attachment point defined in \eqref{opt-d}.

Intuitively, $d^*$ is determined by solving the first order condition $J^\prime(d)=0$ or equivalently 
$h_v(d)=0$. However, the solution to $h_v(d)=0$ may not exist or there could be multiple solutions. We thus use the expression \eqref{opt-d} to address these issues.  If there are multiple solutions to $h_v(d)=0$, the monotonicity of $h(\cdot)$ implies that  the solution set  forms an interval with $d^*$ as defined in \eqref{opt-d} being its  supremum.  Any $d$ in this interval qualifies as an optimal attachment point as they all lead to the same objective function value.  It suffices to use one value to represent them and we select $d^*$, the supremum of the solution set, because CAT bonds with a higher attachment point is generally more attractive to investors. When there is no solution, the optimal attachment point is the point where $h_v(d)$ changes its sign. 


Note that the optimal attachment point $d^*$, as specified by equation \eqref{opt-d}, may be infinite. This corresponds to the strategy of not issuing a CAT bond at all. From a practical standpoint, a natural question arises: under what conditions should the insurer issue a CAT bond, and when should they refrain?  The following proposition address this question. The proof follows directly from Theorem \ref{t:opt-d} and Lemma \ref{l:hprop} and is therefore omitted.

\begin{proposition} \label{p:d*prop}
The optimal attachment point $d^*$ specified by \eqref{opt-d} satisfies the following properties 
\begin{itemize}
\item[(i)] $d^*$ is finite if and only if  
\begin{equation}\label{cond-dfinite}
\frac{\mathbb{E}[v^\prime(X-P_1)]}{\lim\limits_{x\to \infty}v^\prime(x)}  <  \frac{1}{1+\tau_2}  
\end{equation}
\item[(ii)] $d^* \ge \operatorname{VaR}_{\rho_0}(X)$. 
\end{itemize}
\end{proposition}







\begin{remark}\label{r2-d*>0}
Note that for a nonnegative loss random variable $X$, $\operatorname{VaR}_{\rho_0}(X)\ge 0$.   Proposition \ref{p:d*prop}(ii) implies that $d^*\ge 0$ and thus $I^*(0)=0$  for the optimal indemnity $I^*(x) = \frac{(x-d^*)_+}{w} \wedge 1$.  
 It is worth noting that although the conclusion of $d^* \ge \operatorname{VaR}_{\rho_0}(X)$ and thus $d^*\ge 0$ is derived under the assumption that $v$ is increasing, convex, and continuously differentiable, the results $d^*\ge 0$ and hence $I^*(0)=0$ hold true for general convex disutility functions according to Remark \ref{r1-d*>0}. 
\end{remark}



We shall study Problem \eqref{P-EDU'} with the following  disutility functions, 
\begin{equation*}
v_0(x) = x, \quad    v_1(x)=|x|, \quad v_2(x)=x_+, \quad v_3(x)=x^2, \quad v_4(x) = x_+^2.
\end{equation*}

 The function $v_0(x) = x$ represents the baseline risk-neutral scenario and the expected net loss. To model asymmetric risk preferences, we introduce the function $v_2(x) = x_+$ which captures a linear risk aversion and penalizes exclusively positive losses, while $v_4(x) = x^2_+$ represents higher-order risk aversion, penalizing these positive losses more than proportionally. Additionally, we introduce the symmetric disutility functions $v_1(x)$ and $v_3(x)$ to account for general uncertainty aversion, wherein any deviation from zero is penalized regardless of its direction. Specifically, $v_1(x) = |x|$ represents linear uncertainty aversion by treating all deviations equally, whereas $v_3(x) = x^2$ captures quadratic uncertainty aversion, reflecting a strong inherent aversion to variance and volatility. Each of these criteria admits a natural statistical interpretation. The problem induced by $v_1$ is the mean absolute deviation of the net loss, while that induced by $v_3$ is the mean squared value. The problems induced by $v_2$ and $v_4$ are the corresponding partial moments above zero, which in turn serve as downside analogues of the ones previously mentioned. Note that the event $L_d = 0$ marks the break-even state of the sponsor, in which income from premiums and coverages exactly offsets the retained loss and the costs, so that positive deviations represent shortfalls whereas negative deviations from zero correspond to overhedging gains. Each disutility function therefore carries a twofold interpretation, first as an attitude toward risk and second as a statistical summary of the net loss, and the latter reading will play a relevant role for the risk-reduction analysis deferred to Section~\ref{sec:case_study}. As we proceed to the problem solutions, we shall evaluate the appropriateness of these choices by analyzing the behavior of the resulting optimal solutions. To this end, we denote the optimal attachment points under these respective  functions as $d^*_0, d^*_1, d^*_2, d^*_3$, and $d^*_4$.

\begin{proposition}\label{p:d034}
\begin{itemize}
    \item[(i)] The optimal attachment point $d_0^* = \infty$, that is, the optimal strategy is not to issue any CAT bonds under the disutility function $v_0$. 
    \item[(ii)] The optimal attachment points $d_3, d_4$ are respectively unique solutions to the following equations:
    \begin{eqnarray}
  \label{d3eqn}     h_{v_3}(d) &:=& m(d)  -  (1+\tau_2)\mathbb{E}[L_d] = 0 \\
    \label{d4eqn}   h_{v_4}(d) &:=& (m(d))_+ - (1+\tau_2)\mathbb{E}[(L_d)_+] = 0 .
    \end{eqnarray}
\end{itemize}
\end{proposition}

The conclusion of Proposition \ref{p:d034} directly follows from Theorem \ref{t:opt-d} and Proposition \ref{p:d*prop}. Proposition \ref{p:d034}(i) is straightforward. Under a linear disutility function, the objective reduces to minimizing the expected net loss; therefore, the optimal strategy is to avoid issuing CAT bonds entirely to eliminate the associated costs.

To study the optimization problem under disutility functions $v_1$ and $v_2$, it is crucial to determine the sign of $m(d)$. To this end, we make the following assumption to simplify the discussion: 
\begin{equation}\label{m(0)<0}
    (1+\tau_1)\mathbb{E}[X] >  (1+\tau_2)\mathbb{E} [X\wedge w]. 
\end{equation}

Note that $(1+\tau_1)\mathbb{E}[X]$ represents the total premium collected by the insurer, while $(1+\tau_2)\mathbb{E} [X\wedge w] \ge (1+\tau_2)\mathbb{E} [(X-d)_+\wedge w]$ denotes the maximum possible reinsurance premium payable to the CAT bond investors. Intuitively, condition \eqref{m(0)<0} ensures that the outgoing reinsurance premium never exceeds the collected premium, which reflects a reasonable design constraint.

Furthermore, condition \eqref{m(0)<0} implies that $m(0)< 0$. According to Lemma \ref{l:md}, 
$m(d)$ is convex on $[0, \infty)$. Thus, the equation $m(d)=0$ yields only one positive root, denoted by $e_0$. Clearly, $m(d)>0$ if and only if $d>e_0$. Furthermore, define $\eta(d) = d - m(d)$ and note that $\eta$ is strictly increasing, hence invertible. Denote its inverse by $\eta^{-1}$. Note that $\eta^{-1}(e_0) = e_0$ and $L_d > 0$ if and only if $X > \eta(d)$.

\begin{proposition}
\label{p:d1_d2}
    The optimal attachment points, $d^*_1, d^*_2$ are  specified by 
    \begin{eqnarray}
\label{d1-exp}        d^*_1 &=& \eta^{-1}\left( (e_0 \vee \operatorname{VaR}_{\rho^\prime}(X)) \wedge \operatorname{VaR}_{1-\rho^\prime}(X) \right) \\
 \label{d2-exp}       d^*_2 &=&  \eta^{-1}(e_0\vee\operatorname{VaR}_{\rho_0}(X))
    \end{eqnarray}
    where $\rho_0=\frac{\tau_2}{1+\tau_2}$,  $\rho^\prime=\frac{\tau_2}{2(1+\tau_2)}$.
\end{proposition}
\begin{proof}
    With $v_1 =|x|$ and thus $v_1^\prime(x) = \mathbb{I}\{x> 0\} - \mathbb{I}\{x \le 0\}$,
    we have 
    \begin{eqnarray*}
        h_{v_1}(d) & = & \mathbb{I}\{m(d) > 0\} - \mathbb{I}\{m(d) \le  0\}  - (1+\tau_2) \mathbb{E}[\mathbb{I}\{L_d > 0\} - \mathbb{I}\{L_d \le 0\}] \\
         &=& 
         \begin{cases}
2+\tau_2 - 2(1+\tau_2) \mathbb{P}\{X> \eta(d)\}, & d > e_0 \\[4pt]
\tau_2 - 2(1+\tau_2) \mathbb{P}\{X > \eta(d)\},    & 0 \le d \le e_0
\end{cases}
    \end{eqnarray*}

Note that the graph of $h_{v_1}(d)$ has a jump at $e_0$, with 
\begin{equation*}
    h_{v_1}(e_0) = \tau_2 - 2(1+\tau_2) \mathbb{P}\{X > e_0\}, \quad
   h_{v_1}(e_0+) = 2+ \tau_2 - 2(1+\tau_2) \mathbb{P}\{X > e_0\}
\end{equation*}
    This jump divides the graph into two branches and  the position of $h_{v_1}(e_0),0, h_{v_1}(e_0+)$ determines which branch crosses zero.


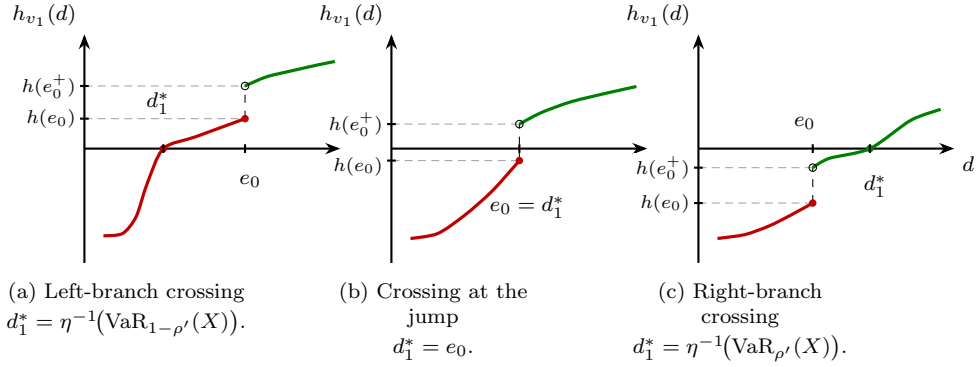
\begin{figure}[htbp]
\centering
\captionsetup[subfigure]{justification=centering,singlelinecheck=false,font=footnotesize}

\begin{subfigure}[t]{0.27\textwidth}
\centering
\begin{tikzpicture}[
    x=0.72cm,y=0.72cm,>=Stealth,
    axis/.style={thick,-{Stealth[length=2mm]}},
    leftcurve/.style={very thick,red!75!black},
    rightcurve/.style={very thick,green!50!black},
    tick/.style={black,thick},
    guide/.style={densely dashed,gray!70},
    ylab/.style={font=\scriptsize,anchor=east,xshift=-2pt,fill=white,inner sep=1pt}
]
    \pgfmathsetmacro{\ds}{1.45}
    \pgfmathsetmacro{\e}{2.95}

    \pgfmathsetmacro{\he}{0.55}
    \pgfmathsetmacro{\hep}{1.15}

    \draw[axis] (0,-2.1) -- (0,2.1);
    \draw[axis] (0,0) -- (4.70,0);

    \node[above left] at (0,2.1) {$h_{v_1}(d)$};
    \node[below right] at (4.70,0) {$d$};

    \draw[guide] (0,\he) -- (\e,\he);
    \draw[guide] (0,\hep) -- (\e,\hep);
    \draw[tick] (0.08,\he) -- (-0.08,\he);
    \draw[tick] (0.08,\hep) -- (-0.08,\hep);
    \node[ylab] at (0,\he) {$h(e_0)$};
    \node[ylab] at (0,\hep) {$h(e_0^+)$};

    \draw[leftcurve]
        plot[smooth] coordinates
        {(0.35,-1.60) (0.70,-1.56) (0.95,-1.25) (1.15,-0.65) (\ds,0.00) (2.05,0.22) (\e,\he)};

    \draw[rightcurve]
        plot[smooth] coordinates
        {(\e,\hep) (3.35,1.32) (3.90,1.45) (4.30,1.54) (4.60,1.60)};

    \draw[dashed] (\e,\he) -- (\e,\hep);

    \fill[red!75!black] (\ds,0.00) circle (1.4pt);
    \fill[red!75!black] (\e,\he) circle (1.4pt);
    \draw (\e,\hep) circle (1.4pt);

    \draw[tick] (\ds,0.08) -- (\ds,-0.08);
    \node[above=10pt,xshift=-2pt] at (\ds,0) {$d_1^*$};

    \draw[tick] (\e,0.08) -- (\e,-0.08);
    \node[below=7pt,xshift=2pt] at (\e,0) {$e_0$};
\end{tikzpicture}
\caption{\centering Left-branch crossing\\[2pt]
$d_1^*=\eta^{-1}\!\bigl(\mathrm{VaR}_{1-\rho'}(X)\bigr)$.}
\end{subfigure}
\hspace{0.055\textwidth}
\begin{subfigure}[t]{0.27\textwidth}
\centering
\begin{tikzpicture}[
    x=0.72cm,y=0.72cm,>=Stealth,
    axis/.style={thick,-{Stealth[length=2mm]}},
    leftcurve/.style={very thick,red!75!black},
    rightcurve/.style={very thick,green!50!black},
    tick/.style={black,thick},
    guide/.style={densely dashed,gray!70},
    ylab/.style={font=\scriptsize,anchor=east,xshift=-2pt,fill=white,inner sep=1pt}
]
    \pgfmathsetmacro{\e}{2.35}

    \pgfmathsetmacro{\he}{-0.22}
    \pgfmathsetmacro{\hep}{0.45}

    \draw[axis] (0,-2.1) -- (0,2.1);
    \draw[axis] (0,0) -- (4.70,0);

    \node[above left] at (0,2.1) {$h_{v_1}(d)$};
    \node[below right] at (4.70,0) {$d$};

    \draw[guide] (0,\he) -- (\e,\he);
    \draw[guide] (0,\hep) -- (\e,\hep);
    \draw[tick] (0.08,\he) -- (-0.08,\he);
    \draw[tick] (0.08,\hep) -- (-0.08,\hep);
    \node[ylab,yshift=-2pt] at (0,\he) {$h(e_0)$};
    \node[ylab] at (0,\hep) {$h(e_0^+)$};

    \draw[leftcurve]
        plot[smooth] coordinates
        {(0.35,-1.65) (0.85,-1.56) (1.35,-1.22) (1.85,-0.78) (\e,\he)};

    \draw[rightcurve]
        plot[smooth] coordinates
        {(\e,\hep) (2.75,0.66) (3.30,0.86) (4.00,1.03) (4.50,1.14)};

    \draw[dashed] (\e,\he) -- (\e,\hep);

    \fill[red!75!black] (\e,\he) circle (1.4pt);
    \draw (\e,\hep) circle (1.4pt);

    \draw[tick] (\e,0.08) -- (\e,-0.08);
    \node[below=15pt,xshift=3pt] at (\e,0) {$e_0=d_1^*$};
\end{tikzpicture}
\caption{\centering Crossing at the jump\\[2pt]
$d_1^*=e_0$.}
\end{subfigure}
\hspace{0.055\textwidth}
\begin{subfigure}[t]{0.27\textwidth}
\centering
\begin{tikzpicture}[
    x=0.72cm,y=0.72cm,>=Stealth,
    axis/.style={thick,-{Stealth[length=2mm]}},
    leftcurve/.style={very thick,red!75!black},
    rightcurve/.style={very thick,green!50!black},
    tick/.style={black,thick},
    guide/.style={densely dashed,gray!70},
    ylab/.style={font=\scriptsize,anchor=east,xshift=-2pt,fill=white,inner sep=1pt}
]
    \pgfmathsetmacro{\e}{2.10}
    \pgfmathsetmacro{\ds}{3.15}

    \pgfmathsetmacro{\he}{-1.00}
    \pgfmathsetmacro{\hep}{-0.35}

    \draw[axis] (0,-2.1) -- (0,2.1);
    \draw[axis] (0,0) -- (4.70,0);

    \node[above left] at (0,2.1) {$h_{v_1}(d)$};
    \node[below right] at (4.70,0) {$d$};

    \draw[guide] (0,\he) -- (\e,\he);
    \draw[guide] (0,\hep) -- (\e,\hep);
    \draw[tick] (0.08,\he) -- (-0.08,\he);
    \draw[tick] (0.08,\hep) -- (-0.08,\hep);
    \node[ylab] at (0,\he) {$h(e_0)$};
    \node[ylab] at (0,\hep) {$h(e_0^+)$};

    \draw[leftcurve]
        plot[smooth] coordinates
        {(0.35,-1.65) (0.85,-1.58) (1.30,-1.42) (1.70,-1.22) (\e,\he)};

    \draw[rightcurve]
        plot[smooth] coordinates
        {(\e,\hep) (2.40,-0.18) (\ds,0.00) (3.90,0.52) (4.45,0.72)};

    \draw[dashed] (\e,\he) -- (\e,\hep);

    \fill[red!75!black] (\e,\he) circle (1.4pt);
    \draw (\e,\hep) circle (1.4pt);
    \fill[green!50!black] (\ds,0.00) circle (1.4pt);

    \draw[tick] (\e,0.08) -- (\e,-0.08);
    \node[below=-15pt,xshift=-3pt] at (\e,0) {$e_0$};

    \draw[tick] (\ds,0.08) -- (\ds,-0.08);
    \node[below=7pt,xshift=2pt] at (\ds,0) {$d_1^*$};
\end{tikzpicture}
\caption{\centering Right-branch crossing\\[2pt]
$d_1^*=\eta^{-1}\!\bigl(\mathrm{VaR}_{\rho'}(X)\bigr)$.}
\end{subfigure}

\caption{Scenarios of $h_{v_1}(d)$ crossing zero to determine $d_1^*$.}
\label{fig:hv1_three_crossings}
\end{figure}


If $h_{v_1}(e_0) > 0$, that is, $e_0 > \operatorname{VaR}_{1-\rho^\prime}(X)$, then the left branch of $h_{v_1}(d)$ crosses zero,  as shown in Figure \ref{fig:hv1_three_crossings}(a). Therefore, $d_1^*$ is obtained by solving
\begin{equation*}
    \tau_2 - 2(1+\tau_2) \mathbb{P}\{X  > \eta(d)\} =0,
\end{equation*}
which yields $d_1^* = \eta^{-1}(\operatorname{VaR}_{1-\rho^\prime}(X))$.

If $h_{v_1}(e_0+) < 0$, that is, $e_0 < \operatorname{VaR}_{\rho^\prime}(X)$, then the right branch of $h_{v_1}(d)$ crosses zero, as shown in Figure \ref{fig:hv1_three_crossings}(c). Therefore, $d_1^*$ is obtained by solving 
\begin{equation*}
    2+\tau_2 - 2(1+\tau_2) \mathbb{P}\{X > \eta(d)\} =0,
\end{equation*}
which yields $d_1^* = \eta^{-1}(\operatorname{VaR}_{\rho^\prime}(X))$.

If $ h_{v_1}(e_0) \le 0 \le h_{v_1}(e_0+)$, that is, $\operatorname{VaR}_{\rho^\prime}(X) \le  e_0 \le  \operatorname{VaR}_{1-\rho^\prime}(X)$, $h_(d)$ crosses zero right at the jump, as shown in Figure \ref{fig:hv1_three_crossings}(b). That means $d_1^* = e_0$, which can be also written as $\eta^{-1}(e_0)$. 

Putting the three pieces together yields \eqref{d1-exp} .

\eqref{d2-exp} can be derived in a similar way.
\end{proof}

Expressions~\eqref{d1-exp} and~\eqref{d2-exp} reveal that $d_1^*=d_2^* = e_0$ if  $\operatorname{VaR}_{\rho_0}(X) \le e_0 \le \operatorname{VaR}_{1-\rho^\prime}(X)$. This corresponds to a wide range of scenarios where both $h_{v_1}(d)$ and $h_{v_2}(d)$ cross zero at their respective jump points, namely $d=e_0$. The definition of $e_0$ implies that the anticipated level of net loss $m(e_0) = 0$. Therefore, within these scenarios, an insurer employing either risk function $v_1$ or $v_2$ should target at a zero  anticipated level of net loss  so as to the minimize the risk exposure. These practical implications  are further illustrated in the case study presented in Section~\ref{sec:case_study}.

\begin{proposition} \label{p:asyvssym}
The optimal attachment points $d^*_1, d^*_2, d^*_3, d^*_4$ can be ordered as follows 
\begin{equation*}
d^*_1 \le d^*_2 \quad  \text{and} \quad  d^*_3 \le d^*_4 
\end{equation*}
\end{proposition}
\begin{proof}
    Since $\rho^\prime \le \rho_0$, then $\operatorname{VaR}_{\rho^\prime}(X)\le \operatorname{VaR}_{\rho_0}(X)$, and thus $(e_0 \vee \operatorname{VaR}_{\rho^\prime}(X)) \wedge \operatorname{VaR}_{1-\rho^\prime}(X) \le e_0 \vee \operatorname{VaR}_{\rho_0}(X)$. Following \eqref{d1-exp} and \eqref{d2-exp}, it implies that $d^*_1 \le d^*_2$ due to the monotonicity of $\eta^{-1}$.   

    To establish the order between $d^*_3$ and $d^*_4$, note that 
\begin{eqnarray*}
h_{v_4}(d^*_4) = 0 
&\Longrightarrow& m(d^*_4)_+ = (1+\tau_2)\mathbb{E}[(L_{d^*_4})_+] \ge 0 \\ 
&\Longrightarrow& m(d^*_4) = m(d^*_4)_+ \ge (1+\tau_2)\mathbb{E}[L_{d^*_4}]
\Longrightarrow h_{v_3}(d^*_4) \ge 0.
\end{eqnarray*}
By Lemma~\ref{l:hprop}(i), $h_{v_3}$ is increasing in $d$. Thus, as the solution to $h_{v_3}(d)=0$, $d^*_3$ must not exceed $d^*_4$.
According to Lemma \ref{l:hprop}(i), $h_{v_3}$ is increasing in $d$. Thus, as the solution to $h_{v_3}(d)=0$, $d_3$ must not exceed $d_4$. 
  
\end{proof}

Proposition \ref{p:asyvssym} indicates that,  when shifting from a symmetric disutility function to its asymmetric counterpart, either from $v_1(x)=|x|$ to $v_2(x) = x_+$ or from $v_3(x) = x^2$ to $v_4(x)=x^2_+$, the decision-maker shifts the emphasis to the tail risk rather than the deviation of either direction. Consequently, it is reasonable to raise the attachment point, thereby triggering the bond later to preserve coverage for more severe losses.

It is worth noting that both $m(d^*_2)$ and $m(d^*_4)$ are nonnegative ($m(d^*_4) \ge 0$ directly follows from \eqref{d4eqn}, while $m(d^*_2) \ge 0$ stems from the condition $d^*_2 \ge e_0$). This nonnegativity is a desirable property; as discussed when $m(d)$ was first introduced, it represents the anticipated level of net loss and should inherently remain nonnegative under any reasonable design. Conversely, $m(d^*_1)$ and $m(d^*_3)$ are not necessarily nonnegative. This suggests that employing symmetric disutility functions to minimize deviations from either side may lead to an overhedging problem.

\subsection{Minimizing spectral risk measure}

In this subsection, we consider Problem \eqref{P-SPEC'}, which is restated below
\begin{equation*}
     \inf_{d \ge -w} J_\phi(d) := \ M_\phi(L_d). 
\end{equation*}
We aim to find the optimal attachment point that minimizes the net loss under certain spectral risk measure. For this purpose, we first introduce some auxiliary notations and results. 

For a given risk spectrum function $\phi$,  define 
\begin{equation*}
\bar{\phi}(u,v) 
= \frac{\int_{u}^v \phi(p) dp}{v-u}, \qquad 0\le u < v \le 1.
\end{equation*}
Recall that $\phi(p)$ represents the decision-weights assigned to the quantile $\operatorname{VaR}_p(X)$ when evaluating the riskiness of $X$. The function $\bar{\phi}(u, v)$ thus represents the average value of $\phi(p)$ over the interval $[u,v]$. 
In particular, $\bar{\phi}(u, 1)$ denotes the average of $\phi$ over $[u,1]$ and is thus referred to as the tail-average. Because $\phi$ is an increasing function, its average naturally increases as the interval $[u, v]$ shifts to the right. Therefore, $\bar\phi(u,v)$ is  increasing in both $u$ and $v$. In particular, we have $\bar\phi(u,v) \le \bar\phi(u,1)$, and $\bar\phi(u,1)$ is increasing in $u$.


\begin{lemma}\label{l:jphi'}
    The first-order derivative of $J_\phi(d)$ with respect to $d$ is given by 
    \begin{equation*}
        J^\prime_\phi(d) = \mathbb{P}\{d<X\le d+w\}\times\Big(\bar\phi\big(F(d), F(d+w)\big)- (1+\tau_2)\Big).
    \end{equation*}
\end{lemma}
\begin{proof}
    Note from \eqref{piecewiseld}, $L_d(X)$ is an increasing function of $X$. According to \cite{McNeil2015}, $\operatorname{VaR}$ is invariant under increasing transformation. Therefore, 
    \begin{equation*}
        \operatorname{VaR}_{p}(L_d(X)) =  L_d(\operatorname{VaR}_{p}(X)). 
    \end{equation*}
   Following \eqref{ld'}, we have 
    \begin{eqnarray*}
        \frac{\partial}{\partial d}  \operatorname{VaR}_{p}(L_d(X)) 
        &=& \frac{\partial}{\partial d} L_d(\operatorname{VaR}_{p}(X)) \\
        &=& m^\prime(d) - 1 + \mathbb{I}\{\operatorname{VaR}_{p}(X)-w < d \le \operatorname{VaR}_{p}(X)\}\\
        &=& -(1+\tau_2)\mathbb{P}\{d<X\le d+w\} + \mathbb{I}\{ F(d) \le p < F(d+w)\}.
    \end{eqnarray*}
    Therefore, 
    \begin{eqnarray*}
        \frac{\partial}{\partial d} J_\phi(d)  
        &=&  \frac{\partial}{\partial d} \int_0^1 \phi(p) \operatorname{VaR}_p(L_d) dp \\ 
        &=&  \int_0^1 \phi(p) \frac{\partial}{\partial d}  \operatorname{VaR}_p(L_d) dp \\ 
        &=&  \int_0^1 \phi(p) \big(-(1+\tau_2)\mathbb{P}\{d<X\le d+w\}\big) dp + 
        \int_0^1 \phi(p) \mathbb{I}\{ F(d) \le p < F(d+w)\} dp \\
        &=& -(1+\tau_2)\mathbb{P}\{d<X\le d+w\} + \int_{F(d)}^{F(d+w)} \phi(p) dp \\
         &=& \mathbb{P}\{d<X\le d+w\}\times\Big(\bar\phi\big(F(d), F(d+w)\big)- (1+\tau_2)\Big).
    \end{eqnarray*}
\end{proof}

\begin{theorem}\label{t:opt-d-phi}
Then the optimal attachment point $d_\phi^*$ that solves Problem \eqref{P-SPEC'} is given by
\begin{equation}\label{opt-d-phi}
d_\phi^*=\inf\{d>-w: \bar\phi\big(F(d), F(d+w)\big) > 1+\tau_2\},
\end{equation}
with the convention $\inf \emptyset = \infty$. 
\end{theorem}
\begin{proof}
 If $\{d : \bar\phi\big(F(d), F(d+w)\big) > 1+\tau_2\} = \emptyset$, that is, $\bar\phi\big(F(d), F(d+w)\big) \le 1+\tau_2$ for all $d > -w$, then so does $J^\prime_\phi(d)$. This implies that the objective function $J_\phi(d)$ is decreasing in $d\in (-w,\infty)$ and thus attains its minimum at $d_\phi^*= \infty = \inf\{d : \bar\phi\big(F(d), F(d+w)\big) > 1+\tau_2\}$.

 Otherwise, $\{d : \bar\phi\big(F(d), F(d+w)\big) > 1+\tau_2\} \neq \emptyset$ and thus $d_\phi^* < \infty$. 
Recall that $\bar\phi(u,v)$ is increasing in $u$ and $v$, thus $\bar\phi\big(F(d), F(d+w)\big)$ is increasing in $d$.
From the definition of $d^*_\phi$, we conclude that  
\begin{equation}\label{barphiiff}
    \bar\phi\big(F(d), F(d+w)\big) > 1+\tau_2 \quad \text{ if and only if } \quad d>  d_\phi^* 
\end{equation}
which in turn implies $J_\phi^\prime(d) >0$ if and only if $d > d_\phi^*$.  
Therefore, $J_\phi(d)$ attains its minimum at $d^*_\phi$. 
 
\end{proof}

\begin{corollary}\label{c:d_phi_infty}
    If $\phi(p)\le 1+\tau_2$ for all $p\in[0,1]$, then 
    $d_\phi^* =\infty$. 
\end{corollary}

As discussed in Section \ref{sec:optimal_indemnity} and reiterated in Subsection \ref{subsec-edu}, a small decrease in the attachment point $d$ expands the CAT bond coverage. Due to the translation invariance and positive homogeneity of a spectral risk measure, this expansion results in a constant marginal cost of $(1+\tau_2)$. Conversely, a small decrease in $d$ shifts the plateau level of the anticipated loss downward, thereby reducing the spectral risk measure. This marginal benefit is quantified by the average decision weight assigned to the losses falling within the coverage layer $[d, d+w]$, equal to $\bar\phi\big(F(d), F(d+w)\big)$.
As long as this marginal benefit exceeds the constant marginal cost, the attachment point should be lowered to minimize the overall risk. This equilibrium yields the optimal attachment point defined in \eqref{opt-d-phi}.

If $\phi(p) \le 1+\tau_2$ for all $p \in [0,1]$, the decision weight assigned to the quantile at any level does not exceed the marginal cost of issuing CAT bond. Consequently, the coverage is economically unjustified. This provides a direct economic interpretation for Corollary \ref{c:d_phi_infty}.



Because the expression for $d_\phi^*$ given by \eqref{opt-d-phi} cannot always be determined explicitly, it is necessary to establish bounding intervals for practical guidance. To this end, we first introduce several auxiliary quantities.

Define 
\begin{eqnarray*}
    \rho_1 &=& \inf\{p\in[0,1]: \bar\phi(p,1)>1+\tau_2\},\\
    \rho_2 &=& \inf\{p\in[0,1]: \phi(p)>1+\tau_2\}. 
\end{eqnarray*}
with the convention of $\inf\emptyset =1$. 


Recall that the quantile level $\rho_0 = \frac{\tau_2}{1+\tau_2}$ was introduced in Section \ref{sec:model-setup}. The following lemma establishes an ordering relationship between $\rho_0, \rho_1,$ and $\rho_2$.

\begin{lemma}\label{l:rhos}
    The quantile levels $\rho_0:=\frac{\tau_2}{1+\tau_2}$, 
    $\rho_1$, and $\rho_2$ satisfy the following inequality:
    \begin{equation*}
        \rho_0 \le \rho_1 \le \rho_2.
    \end{equation*}
\end{lemma}
\begin{proof}
 Following the monotonicity of  $\bar\phi(p,1)$  in $p$ and the definition of $\rho_1$, it holds that $\bar\phi(p,1)>1+\tau_2$ if and only if $p>\rho_1$.
    
    Note that 
    \begin{equation*}
        \bar\phi(\rho_0,1) =  \frac{\int_{\rho_0}^1 \phi(u) du}{1 - \rho_0} 
        \le \frac{\int_{0}^1 \phi(u) du}{1 - \rho_0} = \frac{1}{1- \frac{\tau_2}{1+\tau_2}} = 1+\tau_2.
    \end{equation*}
This implies $\rho_0 \le \rho_1$.

Additionally, the monotonicity of $\phi(p)$ in $p$  implies that 
\begin{equation*}
    \bar\phi(p,1) = \frac{\int_{p}^1 \phi(u) du}{1 - p} \ge \frac{\int_{p}^1 \phi(p) du}{1 - p} = \phi(p).
\end{equation*}
By the definitions of $\rho_1, \rho_2$, it immediately follows that $\rho_1\le \rho_2$. 
 
\end{proof}

\begin{proposition}\label{p:bounds}
 The optimal attachment point $d_\phi^*$ satisfies
\begin{equation*}
    \max\{\operatorname{VaR}_{\rho_1}(X), \operatorname{VaR}_{\rho_2}(X)-w\} \le d_\phi^* \le \operatorname{VaR}_{\rho_2}(X).
\end{equation*} 
\end{proposition}
\begin{proof}
For any $p > \rho_2$, it holds that $\phi(p) >1+\tau_2$,  and thus $\bar\phi(\rho_2, v)  =\frac{\int_{\rho_2}^v \phi(p) dp}{v-\rho_2}>1+\tau_2$ for any $v>\rho_2$. Noting that $F(\operatorname{VaR}_{\rho_2}(X)) = \rho_2$, we have
\begin{eqnarray*}
    \bar\phi\big(F(\operatorname{VaR}_{\rho_2}(X)), F(\operatorname{VaR}_{\rho_2}(X)+w)\big) =     \bar\phi\big(\rho_2, F(\operatorname{VaR}_{\rho_2}(X)+w)\big) > 1+\tau_2,
\end{eqnarray*}
which implies $d_\phi^* \le \operatorname{VaR}_{\rho_2}(X)$ according to \eqref{barphiiff}. 

Similarly, $\bar\phi\big(F(\operatorname{VaR}_{\rho_2}(X)-w), F(\operatorname{VaR}_{\rho_2}(X))\big) \le 1+\tau_2$, implying $d_\phi^* \ge \operatorname{VaR}_{\rho_2}(X)-w$.

Since $\bar\phi(u,v)$ is increasing in $v$, and for any $v>\rho_1$, we have $\bar{\phi}(\rho_1, v) \le \bar{\phi}(\rho_1, 1) \le 1+\tau_2$. Noting that $F(\operatorname{VaR}_{\rho_1}(X)) = \rho_1$, it holds that
\begin{equation*}
 \bar\phi\big(F(\operatorname{VaR}_{\rho_1}(X)), F(\operatorname{VaR}_{\rho_1}(X)+w)\big) \le \bar\phi(\rho_1,1) \le 1+\tau_2,
\end{equation*}
which implies $d_\phi^*\ge \operatorname{VaR}_{\rho_1}(X)$.

\end{proof}

Recall from the interpretation of Theorem \ref{t:opt-d-phi} that determining the optimal attachment point reduces to a comparison of the marginal benefit and marginal cost induced by a decrease in the attachment point. The marginal cost is a constant $(1+\tau_2)$, whereas the marginal benefit is quantified by the average decision weight assigned to the losses covered by the CAT bond, specifically those falling within $[d, d+w]$. Furthermore, $\rho_2$ serves as a break-even point between the point-wise decision weight $\phi(p)$ and the cost $(1+\tau_2)$. Intuitively, under the optimal design, the loss at the quantile corresponding to this break-even point must be covered to achieve equilibrium, which means $\operatorname{VaR}_{\rho_2}(X) \in [d_\phi^*, d_\phi^* +w]$, or equivalently $\operatorname{VaR}_{\rho_2}(X)-w \le d_\phi^* \le \operatorname{VaR}_{\rho_2}(X)$.

Furthermore, $\rho_1$ serves as a break-even point between the tail-average decision weight $\bar\phi(p,1)$ and the cost $(1+\tau_2)$. Because the issuance amount $w$ is finite, the average decision weight assigned to the covered losses never fully encompasses the tail. Consequently, to achieve equilibrium, the optimal attachment point must be  higher than the loss quantile corresponding to $\rho_1$.

It is worth noting that Proposition \ref{p:bounds} implies that $d_\phi^* \ge \operatorname{VaR}_{\rho_1}(X) \ge 0$. This confirms that, within the framework of minimizing the spectral risk measure, the optimal solution endogenously satisfies the practical constraint $I(0)=0$, as discussed at the end of Section \ref{sec:optimal_indemnity}.



\begin{proposition} \label{p:d_phi_com}
     $d_\phi^*$ is independent of $\tau_1$, increasing in $\tau_2$, and decreasing in $w$.
\end{proposition}
\begin{proof} Recall that $d_\phi^*$ is the point where $\bar\phi\big(F(d), F(d+w)\big)$, as a function of $d$, crosses the constant level $(1+\tau_2)$. 

Because both $\bar\phi\big(F(d), F(d+w)\big)$ and $1+\tau_2$ independent of $\tau_1$, so is $d_\phi^*$.

Since $\bar\phi(u,v)$ is increasing in $u$ and $v$, then $\bar\phi\big(F(d), F(d+w)\big)$ is increasing in $d$ (as well as in $w$), this means the crossing occurs once and from below. 

An increase in $\tau_2$ raises the level $1+\tau_2$ that $\bar\phi\big(F(d), F(d+w)\big)$ must cross, which delays the crossing and results in an increase in $d_\phi^*$. Similarly, an increase in $w$ shifts the graph of $\bar\phi\big(F(d), F(d+w)\big)$ as a function of $d$ upward and thus makes the crossing  earlier, resulting in a decrease in $d_\phi^*$. 
\end{proof}

Proposition \ref{p:d_phi_com} establishes comparative statics that characterize the behavior of the optimal attachment point in response to changes in the parameters. Intuitively, this can be interpreted as follows.
Due to the translation invariance property of spectral risk measures, the insurance premium $P_1$, and consequently $\tau_1$, contributes to the objective function solely as an additive constant. Because it does not interact with other terms, it does not affect the determination of the optimal attachment point. Conversely, an increase in $\tau_2$ implies that issuing CAT bonds becomes more expensive. A rational decision-maker would therefore reserve this costly coverage for more severe, albeit less likely, losses. This corresponds to delaying the CAT bond trigger, which inherently increases the optimal attachment point. Furthermore, a higher $w$ implies greater coverage capacity. With extreme tail risks better mitigated, the decision-maker is incentivized to shift the coverage layer downward to absorb middle-tranche losses.

\begin{proposition}\label{p:opt_d_beta}
  Consider Problem \eqref{P-SPEC'} where $M_\phi$ is the $\operatorname{TVaR}_\beta$ measure, defined by the risk spectrum function $\phi(p)=\frac{\mathbb{I}\{p>\beta\}}{1-\beta}$. Under this setup, the optimal attachment point $d_\beta^*$ is given by
\begin{itemize}
    \item[(i)] If $\beta \le \rho_0:= \frac{\tau_2}{1+\tau_2}$, then $d_\beta^* = \infty$;
    \item[(ii)] If $\beta > \rho_0 := \frac{\tau_2}{1+\tau_2}$, then
 $d_\beta^*$ is the unique solution to  
\begin{equation}\label{opt_d_beta}
 \frac{1 -  \frac{1-F(d+w)}{1-\beta}}{F(d+w)-F(d)} = 1+\tau_2
\end{equation}
Furthermore, $d_\beta^*$ is increasing in $\beta$ and 
$$d_\beta^* \in [\max\{\operatorname{VaR}_{\rho_0}(X), \operatorname{VaR}_\beta(X)-w\}, \operatorname{VaR}_\beta(X)].$$
\end{itemize}
\end{proposition}
\begin{proof}
    Note that $\phi(p)$ is a step function taking two values: $0$ and $\frac{1}{1-\beta}$. It is easy to derive that
    \begin{equation*}
  \bar\phi(u,v) =
  \begin{cases}
  0 & \quad 0 \le u<v \le \beta;\\
  \frac{1-\frac{1-v}{1-\beta}}{v-u} & \quad 0\le u < \beta < v \le 1; \\
  \frac{1}{1-\beta} & \quad \beta\le u<v \le 1.
  \end{cases}   
    \end{equation*}

   (i) If $\beta \le \frac{\tau_2}{1+\tau_2}$, then $\frac{1}{1-\beta} \le 1+\tau_2$.  Thus $\bar\phi(u,v)$ would never exceed $ 1+\tau_2$ for any $u, v\in[0,1]$. According to Theorem \ref{t:opt-d-phi}, $d_\beta^* = \infty$. 

   (ii) If $\beta > \frac{\tau_2}{1+\tau_2}$, then $\frac{1}{1-\beta} > 1+\tau_2$. Thus, $\{d: \bar\phi\big(F(d), F(d+w)\big)>1+\tau_2\} \neq \emptyset$, which implies that $d_\beta^*<\infty$. Noting that $\bar\phi\big(F(d), F(d+w)\big)$ is continuous in $d$, finding $d_\beta^*$ reduces to solving the equation $\bar\phi\big(F(d), F(d+w)\big) = 1+\tau_2$, which is equivalent to \eqref{opt_d_beta}.   

Note that the left-hand side of \eqref{opt_d_beta} is decreasing in $\beta$ and increasing in $d$. The latter property holds because $F(d+w)-F(d)$ is a decreasing function of $d$, due to the assumption that $X$ has a decreasing probability density function. Therefore, treating $d_\beta^*$ as an implicit function of $\beta$ determined by \eqref{opt_d_beta}, we conclude that $d_\beta^*$ is increasing in $\beta$.

   The bounds for $d_\beta^*$ follow from Proposition \ref{p:bounds} by noting that $\rho_2 = \beta$.
   
\end{proof}

 If $\beta=\frac{\tau_2}{1+\tau_2}$, according to Proposition \ref{p:opt_d_beta}, it is straightforward to conclude that $d_\beta^*=\infty$. 
However, it is worth noting that in this case, the optimal solution is not unique. Specifically, if  $\beta=\frac{\tau_2}{1+\tau_2}=\rho_0$, then  $\phi(p) = (1+\tau_2)\mathbb{I}\{p>\beta\}$, and thus $\bar\phi(u,v)=1+\tau_2$ for any $\rho_0\le u<v\le 1$. Therefore, for any $d\ge \operatorname{VaR}_{\rho_0}(X)$, 
\[
J_\phi^\prime(d) = \mathbb{P}\{d<X \le d+w\}\Big(\bar\phi\big(F(d), F(d+w)\big) - (1+\tau_2)\Big) = 0.
\]
This means that the objective function $J_\phi(d)$ stays constant on the interval $[\operatorname{VaR}_{\rho_0}(X), \infty)$. In this sense, any $d$ falling in this interval serves as an optimal attachment point as they all result in identical minimum spectral risk measure. By convention, we select $d_\beta^*=\infty$ as the representative solution, as it corresponds to a no-issuance strategy and inherently eliminates administrative complexities.

The investigation of this case also provides insights into how the risk spectrum function $\phi(p)$ impacts decision-making. In this particular scenario, the decision weight assigned to any loss greater than $\operatorname{VaR}_{\rho_0}(X)$ is a constant $1+\tau_2$. Because this exactly equals the marginal cost of issuing a CAT bond, the decision-maker is economically indifferent to whether or not to cover losses at these levels.

The monotonic behavior of $d_\beta^*$ with respect to $\beta$ in Proposition \ref{p:opt_d_beta} admits a clear economic interpretation. Because $\beta$ represents the confidence level of the TVaR measure, a higher $\beta$ implies that the decision-maker is concerned with more extreme tail losses. Consequently, it is rational to raise the optimal attachment point to reserve coverage for these highly severe events.

Notably, this result serves as a meaningful supplement to the comparative statics established in Proposition \ref{p:d_phi_com}. Specifically, Proposition \ref{p:d_phi_com} characterizes the behavior of the optimal attachment point $d_\phi^*$ in response to changes in the real-valued parameters $\tau_1$, $\tau_2$, and $w$. In contrast, the present result addresses the functional parameter: the risk spectrum function $\phi$. Admittedly, when restricted to the TVaR measure, this functional parameter $\phi$ is entirely characterized by the real-valued parameter $\beta$. Nevertheless, this result motivates a broader inquiry into how the optimal attachment point responds to a general variation in the risk spectrum function, which remains an open question for future research.

The following proposition presents the optimal solution to Problem \eqref{P-SPEC'} under the proportional hazard transform and establishes its range. The proof follows directly from Theorem \ref{t:opt-d-phi} and Proposition \ref{p:bounds} and is therefore omitted.

\begin{proposition}\label{p:opt_d_r}
 Consider Problem \eqref{P-SPEC'} where $M_\phi$ is a proportional hazard transform, characterized by the risk spectrum function $\phi(p)=r(1-p)^{r-1}$ with parameter $r\in(0,1)$. Under this setup, the optimal attachment point $d_r^*$ is the unique solution to
 \begin{equation}\label{opt_d_r}
     \frac{\big(S(d)\big)^r - \big(S(d+w)\big)^r}{S(d)-S(d+w)} = 1+\tau_2,
 \end{equation}
 where $S(x) = \mathbb{P}\{X>x\}$ is the survival function of $X$.
 
Furthermore,  
$$d_r^* \in [\max\{\operatorname{VaR}_{\rho_1}(X), \operatorname{VaR}_{\rho_2}(X)-w\}, \operatorname{VaR}_{\rho_2}(X)]$$
with 
$\rho_1 = 1 - \big(\frac{1}{1+\tau_2}\big)^{\frac{1}{1-r}}$ and 
$\rho_2 = 1 - \big(\frac{r}{1+\tau_2}\big)^{\frac{1}{1-r}}$.
\end{proposition}

\section{Case study: U.S. severe storm losses}
\label{sec:case_study}

To illustrate the theoretical results developed in Section~\ref{sec-optimal_attachment_point}, we conduct a case study to analyze annual U.S. severe storm losses. The dataset is sourced from the U.S. Billion-Dollar Weather and Climate Disasters database, formerly administered by the National Oceanic and Atmospheric Administration's National Centers for Environmental Information (NOAA NCEI); the database was retired in 2025, but its records remain archived by NCEI \citep{NOAA2024}. It comprises the annual aggregate direct economic losses attributed to severe storms over the period 1980--2024, CPI-adjusted and expressed in billions of U.S. dollars. Because the database exclusively records events with losses of at least one billion dollars, the data are left-censored at this threshold: of the 45 years, $n=40$ have exact observations and five are known only to be below one billion dollars. The complete dataset is presented in Table~\ref{tab:severe_storm_data}. Key characteristics of the exact observations are summarized in Table~\ref{tab:desc_severe_storms} and illustrated in Figure~\ref{fig:severe_storm_visualizations}.

\begin{table}[!htbp]
\centering
\footnotesize
\begin{tabular}{rrrrrr}
\toprule
Year & Loss & Year & Loss & Year & Loss \\
\midrule
1980 & $<1.0$  & 1995 & 12.6 & 2010 & 13.0 \\
1981 & 1.4  & 1996 & $<1.0$  & 2011 & 45.1 \\
1982 & 3.2  & 1997 & 1.9  & 2012 & 22.5 \\
1983 & $<1.0$  & 1998 & 9.1  & 2013 & 14.5 \\
1984 & 3.3  & 1999 & 3.8  & 2014 & 15.8 \\
1985 & 1.8  & 2000 & 1.3  & 2015 & 10.9 \\
1986 & 1.5  & 2001 & 6.7  & 2016 & 19.2 \\
1987 & $<1.0$  & 2002 & 5.0  & 2017 & 24.7 \\
1988 & $<1.0$  & 2003 & 13.8 & 2018 & 17.0 \\
1989 & 1.4  & 2004 & 2.9  & 2019 & 17.1 \\
1990 & 2.0  & 2005 & 1.4  & 2020 & 40.7 \\
1991 & 1.5  & 2006 & 11.0 & 2021 & 23.4 \\
1992 & 7.2  & 2007 & 3.8  & 2022 & 22.8 \\
1993 & 1.4  & 2008 & 13.3 & 2023 & 55.7 \\
1994 & 2.1  & 2009 & 11.8 & 2024 & 46.8 \\
\bottomrule
\end{tabular}
\caption{Annual U.S. severe storm losses, 1980-2024 (CPI-adjusted, in \$billion)}
\label{tab:severe_storm_data}
\end{table}

\begin{table}[!htbp]
\centering
\small
\renewcommand{\arraystretch}{1.2}
\begin{tabular}{l r @{\hspace{1.5cm}} l r}
\toprule
\textbf{Metric} & \textbf{Value} & \textbf{Metric} & \textbf{Value} \\
\midrule
Number of Observations        & 40     & 10th Percentile               & 1.40   \\
Mean                          & 12.86  & 25th Percentile (Q1)          & 2.08   \\
Median                        & 10.00  & 75th Percentile (Q3)          & 17.02  \\
Standard Deviation            & 13.71  & 90th Percentile               & 26.30  \\
Coefficient of Variation (\%) & 106.59 & Interquartile Range (IQR)     & 14.95  \\
Minimum                       & 1.30   & Skewness                      & 1.58   \\
Maximum                       & 55.70  & Excess Kurtosis                      & 1.96   \\
\bottomrule
\end{tabular}
\caption{Descriptive statistics of U.S. severe storm losses, 1980--2024.}
\label{tab:desc_severe_storms}
\end{table}

\begin{figure}[!htbp]
    \centering
    \includegraphics[width=\textwidth]{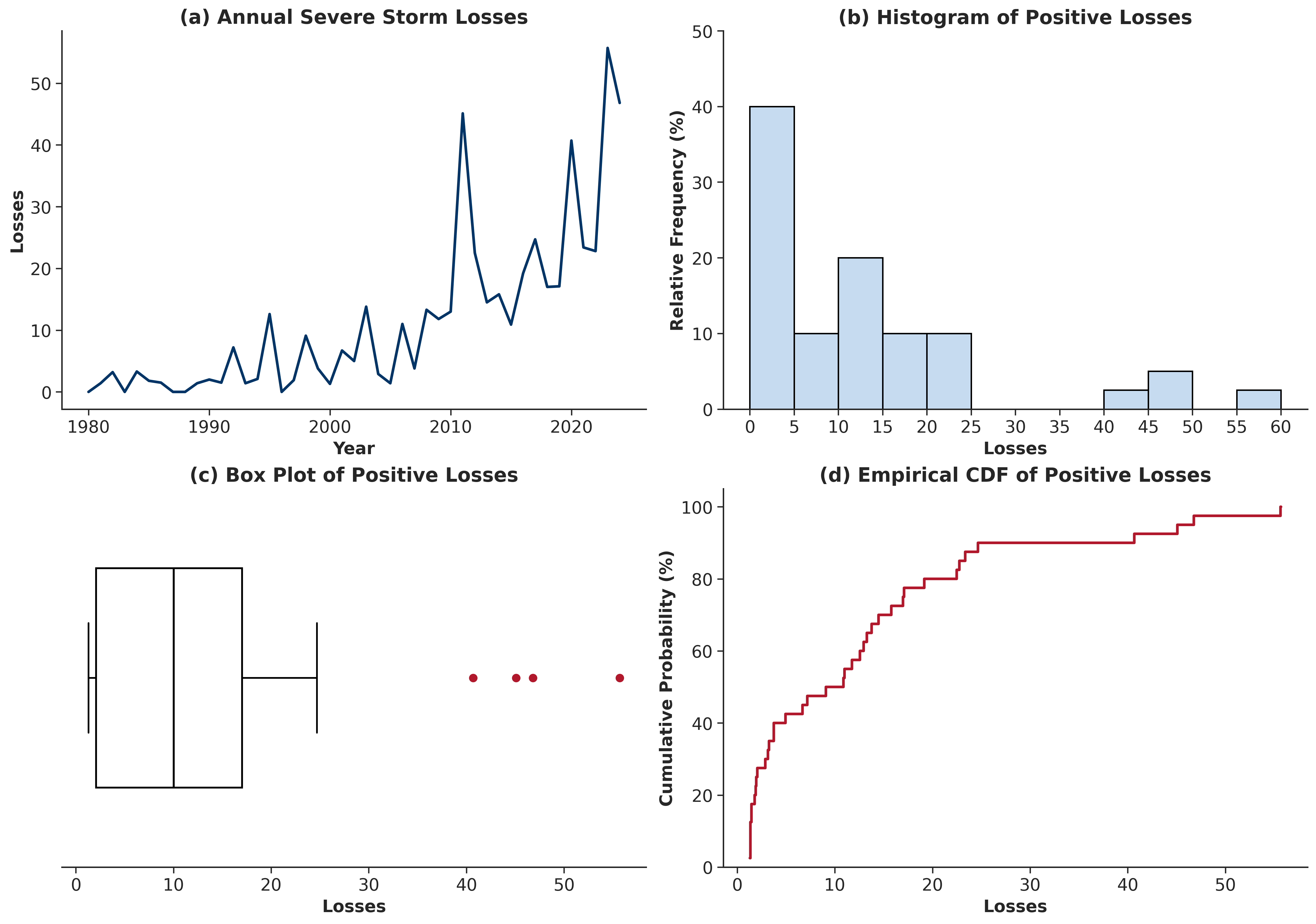}
    \caption{Visual summary of U.S. severe storm losses, 1980--2024.}
    \label{fig:severe_storm_visualizations}
\end{figure}

\FloatBarrier


As shown in Figure~\ref{fig:severe_storm_visualizations}(b) and (d), the empirical distribution is strongly right-skewed, with mass concentrated at moderate loss levels and a pronounced upper tail. The boxplot in Figure~\ref{fig:severe_storm_visualizations}(c) further shows high-end outliers alongside a tightly concentrated body. Together with the statistics in Table~\ref{tab:desc_severe_storms}, these patterns indicate that severe storm losses exhibit catastrophic characteristics and should be modeled by a heavy-tailed distribution. In addition, Figure~\ref{fig:severe_storm_visualizations}(a) shows rising levels and substantial volatility over time, with episodic extreme spikes, consistent with observed climate trends and growing exposure.

\subsection{Parametric fitting of the loss distribution}
\label{sec:parametric-fit}

In this subsection, we fit the loss data to parametric models and implement the optimal CAT bond design derived in Sections 3 and 4. We consider three candidate distributions: Lomax (Pareto Type II), Weibull, and lognormal, as they are standard in actuarial science for modeling heavy-tailed losses \citep{klugman2012loss, embrechts1997modelling}. We estimate the parameters by maximum likelihood for left-censored data as in   \citep{klugman2012loss} and assess goodness-of-fit with the statistics summarized in Table~\ref{tab:model-comparison}.




\begin{table}[H]
\centering
\small
\setlength{\tabcolsep}{5pt}
\renewcommand{\arraystretch}{1.2}
\begin{tabular}{l c c S[table-format=-3.2] S[table-format=3.2] S[table-format=3.2] S[table-format=1.2] S[table-format=1.3] }
\toprule
& \multicolumn{2}{c}{Estimates} & & & & &  \\
\cmidrule(lr){2-3} 
Model & Shape & Scale & {Log-lik.} & {AIC} & {BIC} & {$\Delta$AIC} & {$\omega_i$} \\
\midrule
Weibull   & $\hat k = 0.80$     & $\hat\lambda = 10.16$ & -153.05 & 310.10 & 313.71 & 0.00 & 0.554  \\
Lomax     & $\hat\alpha = 3.44$ & $\hat\theta = 28.92$  & -153.90 & 311.80 & 315.42 & 1.70 & 0.237  \\
Lognormal & $\hat\sigma = 1.39$  & $\hat\mu = 1.67$       & -154.02 & 312.04 & 315.65 & 1.94 & 0.210 \\
\bottomrule
\end{tabular}
\caption{Model fitting statistics}
\label{tab:model-comparison}
\end{table}

 As shown in Table~\ref{tab:model-comparison}, the Weibull model yields the highest log-likelihood and lowest information criteria, though its performance advantage is marginal. Because the Lomax and lognormal models fall within the $\Delta AIC<2$ threshold relative to the Weibull and the Akaike weights 
  fail to assign any model a decisive share of evidence, all three specifications retain substantial empirical support \citep{anderson2004model}. Furthermore, the Kolmogorov--Smirnov test does not reject any of the models at the $\gamma=0.05$ significance level. Overall, empirical diagnostics support all three candidate distributions without decisively preferring any single  model.

 Among the three candidate distributions, only the Lomax exhibits regular variation in the sense of \citet{bingham1989regular} 
and thus assigns greater probability mass to extreme events. Therefore, we select it 
 as a more conservative choice for modeling catastrophic risk. Furthermore, among the three candidates only the Lomax density is decreasing, as required by the model in Section~\ref{sec:model-setup}.
As an additional mathematical advantage, the Lomax model allows closed-form expressions for the key loss functionals in the subsequent analysis.

\subsection{Numerical implementation of the model}
\label{sec:numerical_implementation}

Using the loss distribution calibrated in Section~\ref{sec:parametric-fit}, we compute the optimal attachment point $d^*$ under each risk assessment criterion. Following Theorem~\ref{thm:optform} and the results derived in Section~\ref{sec-optimal_attachment_point}, determining $d^*$ reduces to solving the corresponding first-order conditions. We shall first evaluate the optimal attachment point under the baseline parameter configuration and then perform a comparative static analysis by varying the underlying parameters.


 The baseline parameter configuration is summarized in Table~\ref{tab:baseline_config}. The distributional parameters $\hat{\alpha}$ and $\hat{\theta}$ are the maximum-likelihood estimates for the Lomax model obtained in Section~\ref{sec:parametric-fit}, and the bond volume $w$ is set to $2\mathbb{E}[X]$. Note that under this configuration, condition~\eqref{m(0)<0} holds. Table~\ref{tab:baseline_results} reports the optimal attachment points $d^*$ across all six risk criteria, in terms of both absolute values and  relative magnitudes normalized by the expected loss, alongside the corresponding theoretical bounds for $d^*$ and the relative reduction of the objective functions. 
Figure \ref{fig:optimal_attachment_points} presents the graphs of the optimal indemnity functions across all risk criteria.

The baseline parameter configuration is summarized in Table~\ref{tab:baseline_config}. The distributional parameters $\hat{\alpha}$ and $\hat{\theta}$ are the maximum-likelihood estimates for the Lomax model obtained in Section~\ref{sec:parametric-fit}, and the bond volume $w$ is set to $2\mathbb{E}[X]$. Note that under this configuration, condition~\eqref{m(0)<0} holds. Table~\ref{tab:baseline_results} reports the optimal attachment points $d^*$ across all six risk criteria, in terms of both absolute values and relative magnitudes normalized by the expected loss, alongside the corresponding theoretical bounds for $d^*$ and the relative reduction of the objective functions. 
 The objective functions under the expected disutility framework can be decomposed  into probabilities and (partial) moments of $X$, the computation of which is straightforward. Under the spectral risk measure framework, the computation is based on the expression of $\operatorname{VaR}_\beta(L_d)$ 
 provided in 
 Lemma~\ref{lemma:netloss-dist} in Appendix~\ref{app:netloss}.
 Figure~\ref{fig:optimal_attachment_points} presents the graphs of the optimal indemnity functions across all risk criteria.


\begin{table}[!htbp]
\centering
\small
\begin{tabular}{llc}
\toprule
& Parameter & Value \\
\midrule
\multirow{2}{*}{Loss distribution (Pareto--Lomax)}
  &  shape $\hat{\alpha}$ & $3.44$ \\
  &  scale $\hat{\theta}$ & $28.92$ \\
\midrule
\multirow{3}{*}{Pricing and contract parameters}
  & Premium loading $\tau_1$ & $0.1$ \\
  & Investor loading $\tau_2$ & $0.3$ \\
  & Bond issuance volume $w$ & $23.71$ \\
\midrule
\multirow{2}{*}{Spectral risk measure parameters}
  & TVaR confidence level $\beta$ & $0.95$ \\
  & PHT distortion parameter $r$ & $\frac{1}{3}$ \\
\bottomrule
\end{tabular}
\caption{Baseline parameter configuration.}
\label{tab:baseline_config}
\end{table}

\begin{table}[!htbp]
\centering
\small
\begin{tabular}{lcccc}
\toprule
Risk criterion & $d^*$ & $d^*/\mathbb{E}[X]$ (\%) & Risk reduction (\%) & Ranges \\
\midrule
$\operatorname{TVaR}_{\beta}$ & $20.02$ & $168.84$ & $37.46$ & $[16.46,\ 40.18]$ \\
$\operatorname{PHT}_{r}$      & $14.74$ & $124.33$ & $4.21$  & $[3.51,\ 23.44]$ \\
$v_4(x)=x_+^2$ & $12.59$ & $106.18$ & $40.86$ & $[\,2.29,\ \infty)$ \\
$v_3(x)=x^2$   & $6.98$  & $58.89$  & $46.95$ & $[\,0,\ \infty)$ \\
$v_2(x)=x_+$   & $6.14$  & $51.75$  & $56.03$ & $[\,2.29,\ \infty)$ \\
$v_1(x)=|x|$   & $6.14$  & $51.75$  & $64.67$ & $[\,0,\ \infty)$ \\
\bottomrule
\end{tabular}
\caption{Optimal $d^*$ and relevant quantities under the baseline configuration. }
\label{tab:baseline_results}
\end{table}


\FloatBarrier

\begin{figure}[!htbp]
\centering
\includegraphics[width=0.85\textwidth]{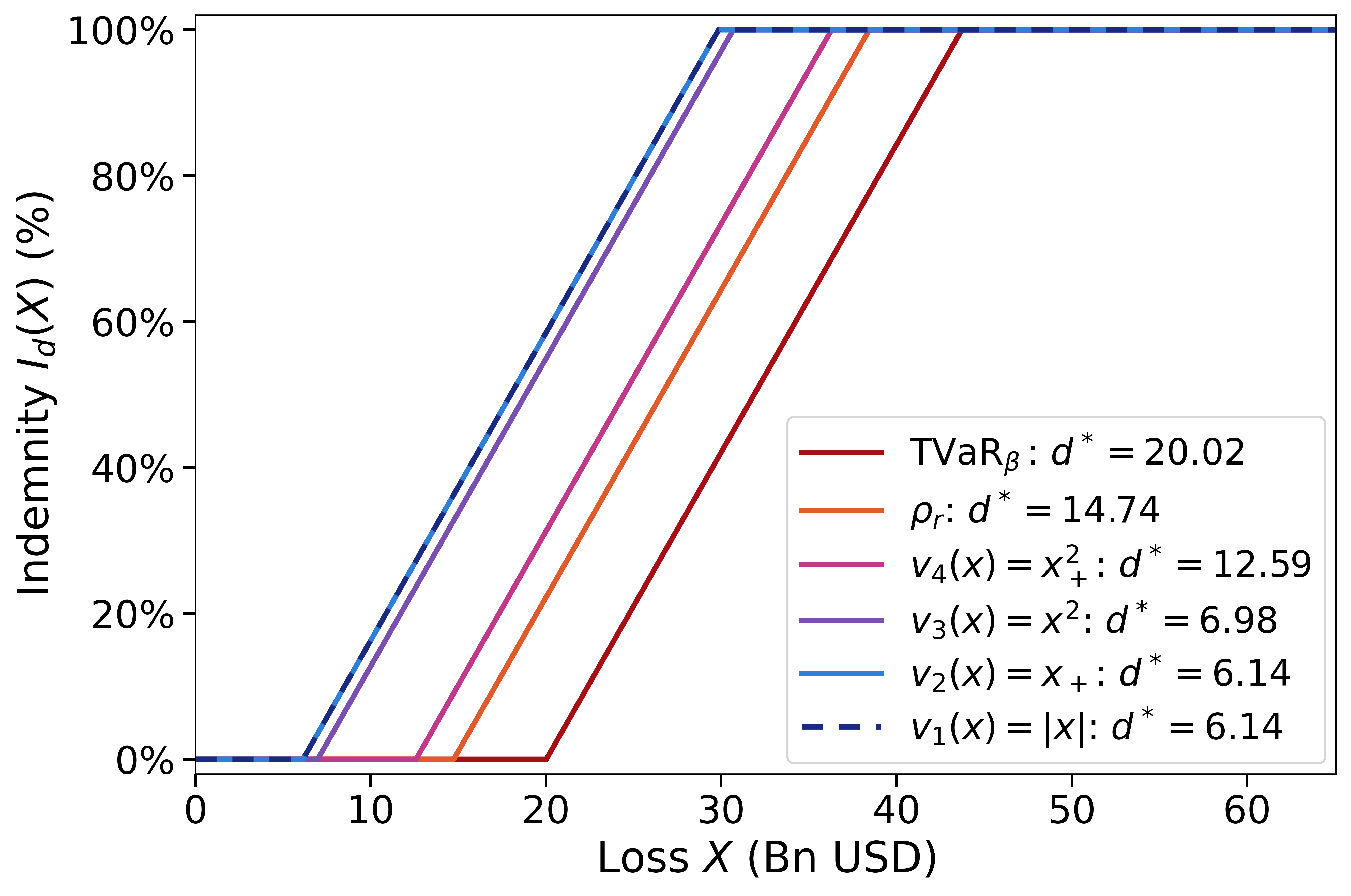}

\caption{
Optimal indemnity functions. 
}

\label{fig:optimal_attachment_points}
\end{figure}

\FloatBarrier



 Table~\ref{tab:baseline_results} reveals that, across all risk criteria, the optimal attachment points range from 50\% to 170\% of the expected value of original losses. In addition, all these attachment points strictly fall within their theoretical ranges established in Proposition~\ref{p:opt_d_beta}, Proposition~\ref{p:opt_d_r} and Proposition~\ref{p:d*prop}, providing a robust numerical consistency check. It is worth noting that the optimal attachment points for the expected disutility models under $v_1 = \vert{}x\vert{}$ and $v_2 = x_+$ coincide. This alignment is consistent with the findings established in Proposition~\ref{p:d1_d2}, which provides theoretical evidence that both models share a common optimal attachment point across a wide range of scenarios.  Furthermore, the risk-reduction column reports the relative decrease in each objective function when comparing the optimal CAT bond coverage to a scenario with no coverage ($d=d^*$ vs. $d=\infty$), serving as an intuitive measure of the CAT bond's efficiency in reducing risk exposure. The results demonstrate that for most models, the optimal CAT bond coverage yields a significant risk reduction. The sole exception is the PHT measure with $r=1/3$; because this distortion index is close to $1/\hat{\alpha}$, the PHT assigns substantial decision weight to the extreme tail, which cannot be effectively mitigated by CAT bond with a coverage layer of only $w=2\mathbb{E}[X]$.

 Beyond the baseline parameter configuration, we vary the parameters and investigate how the behavior of the optimal attachment point in response to these changes. The results are visualized in Figure~\ref{fig:comparative_statics}.

\begin{figure}[!htbp]
\centering



\begin{subfigure}[t]{0.48\textwidth}
    \centering
    \includegraphics[width=\textwidth]{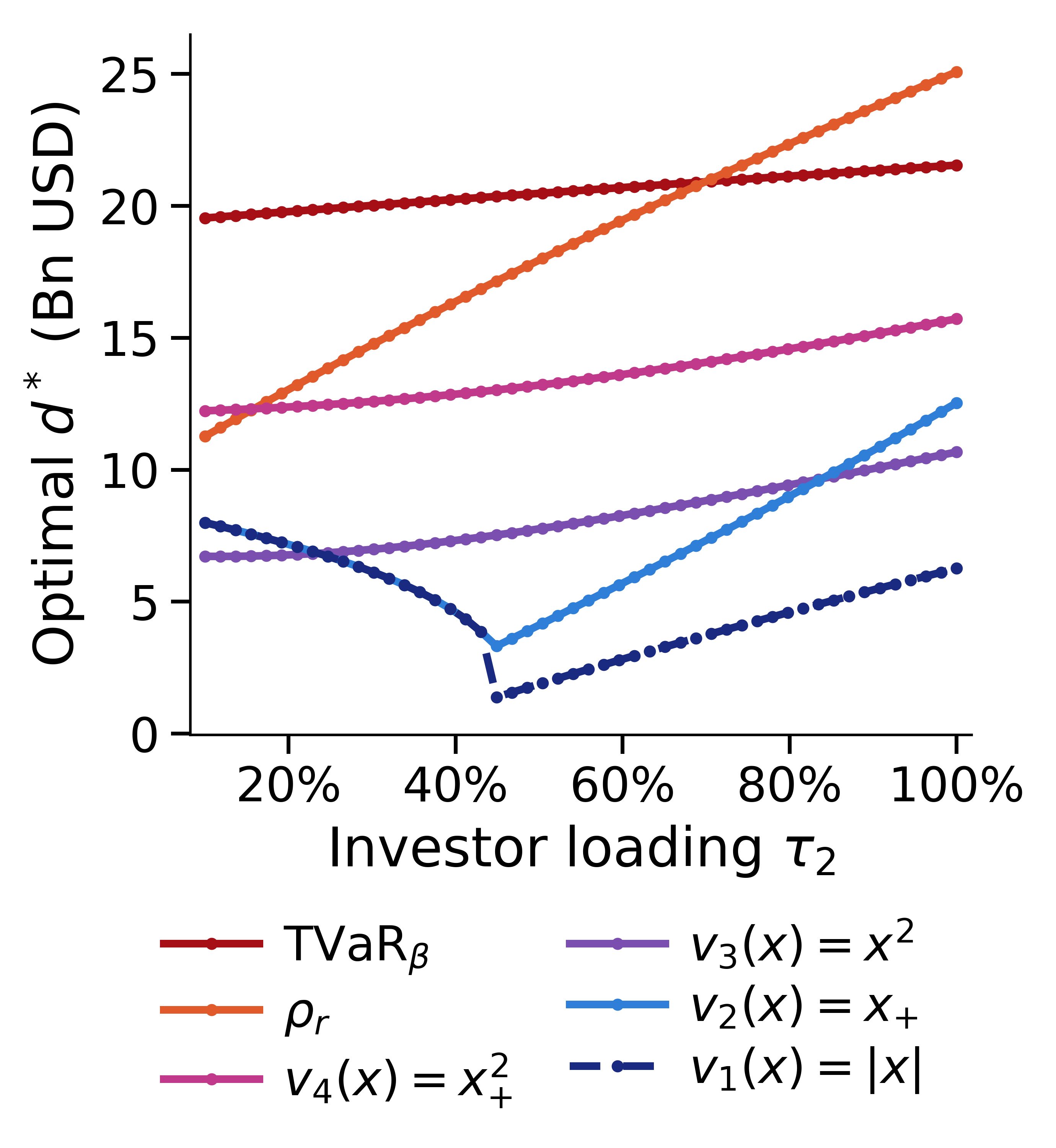}
    \caption{Comparative statics with respect to $\tau_2$.}
    \label{fig:cs_tau2}
\end{subfigure}
\begin{subfigure}[t]{0.48\textwidth}
    \centering
    \includegraphics[width=\textwidth]{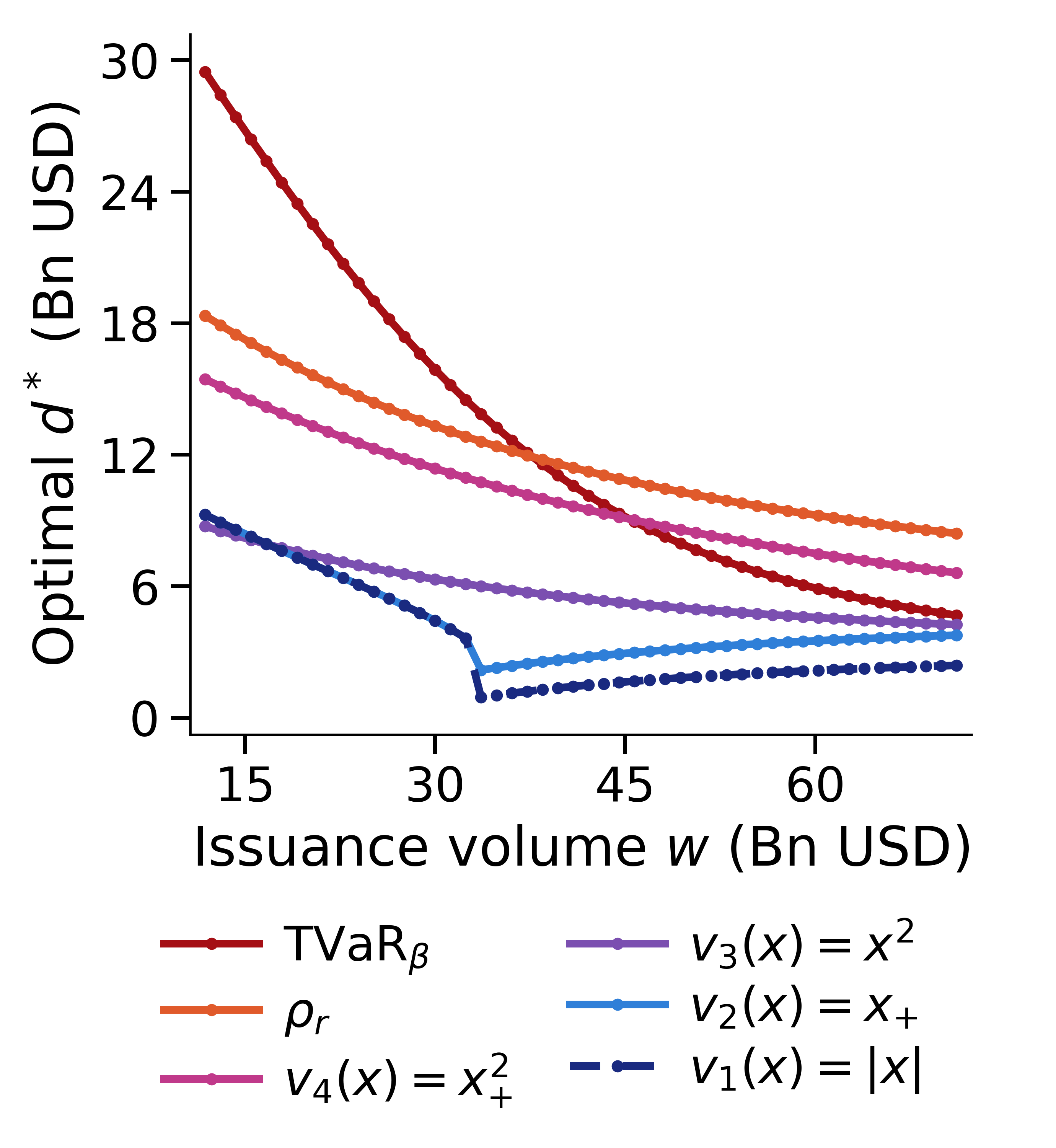}
    \caption{Comparative statics with respect to $w$.}
    \label{fig:cs_w}
\end{subfigure}

\caption{Comparative statics of  $d^*$ with respect to  $\tau_2$, and $w$.}
\label{fig:comparative_statics}
\end{figure}

\FloatBarrier

 Across the two panels in Figure~\ref{fig:comparative_statics}, the spectral risk measure minimizers behave exactly as Proposition~\ref{p:d_phi_com} outlines: they are  increasing in the CAT bond pricing loading $\tau_2$ and decreasing in the CAT bond issue volume $w$. These graphical results provide a numerical validation of the theoretical results. The interpretations of these behaviors  have been discussed following Proposition~\ref{p:d_phi_com}  and are thus omitted here.

 Within the expected disutility framework, no theoretical comparative static results have been established due to its analytical complexity. However, in the specific setting of our case study,  the expected disutility minimizers exhibit behaviors that closely align with those of the spectral risk measure minimizers and yield similarly intuitive interpretations. Specifically, for the square and semi-square disutility functions ($v_3(x)=x^2$ and $v_4(x) = x_+^2$), the minimizers are increasing in $\tau_2$ and decreasing in $w$. This perfectly echoes the economic intuition established under the spectral framework: higher pricing loadings compel the decision-maker to reserve costly coverage for more severe tail events, whereas greater issuance volume incentivizes shifting coverage downward to absorb middle-tranche losses.

 For the two semi-linear disutility functions, $v_1(x)=|x|$ and $v_2(x)=x_+$, the corresponding optimal attachment points, $d_1^*$ and $d_2^*$, as functions of $\tau_2$ and $w$, exhibit the same monotonicity behavior as two other disutility functions only within a restricted range. Specifically, both $d_1^*$ and $d_2^*$ increase with $\tau_2$ once $\tau_2$ exceeds a threshold value, and decrease with $w$ until $w$ reaches a threshold. The mechanisms underlying the divergence of $d_1^*$ and $d_2^*$ from $d_3^*$ and $d_4^*$ outside this parameter region remain to be investigated.    
Additionally, it is worth noting that the $d^*$ curves for these two disutility functions overlap significantly. This visually confirms that $d_1^*$ and $d_2^*$ coincide across a wide range of scenarios, as presented in Proposition~\ref{p:d1_d2}.

\begin{figure}[H]
    \centering
    \includegraphics[width=0.78\textwidth]{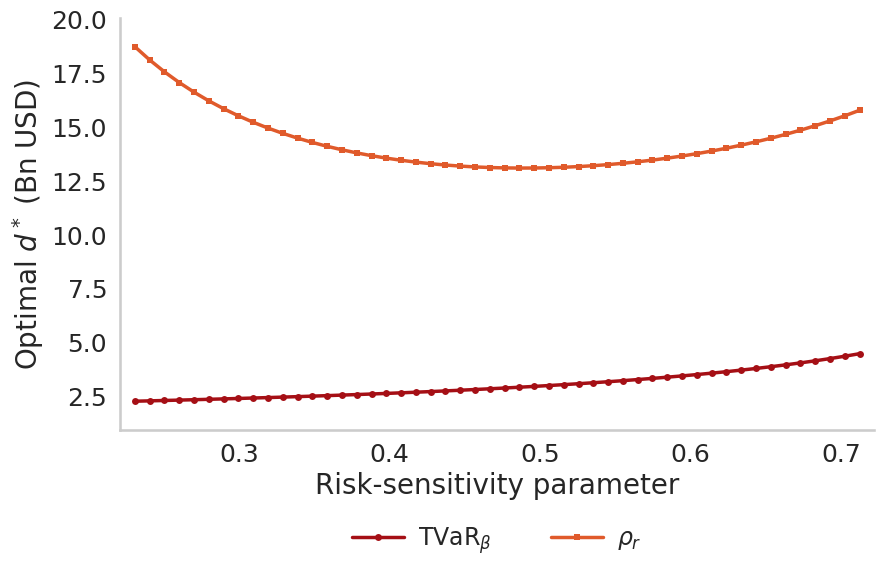}
    \caption{Comparative statics of  $d^*$ w.r.t the parameters of risk spectrum functions
    }
    \label{fig:comparative_statics_tvar_pht}
\end{figure}

 Figure~\ref{fig:comparative_statics_tvar_pht} illustrates how the minimizers of the spectral risk measures vary with their respective risk-sensitivity parameters. The blue curve represents the minimizer of $\operatorname{TVaR}_\beta$, denoted by $d_\beta^*$, as a function of the confidence level $\beta$. Consistent with Proposition~\ref{p:opt_d_beta}, $d_\beta^*$ increases monotonically with $\beta$. The purple curve depicts $d_{1-r}^*$, the minimizer associated with the proportional hazard transform (PHT) of the net loss, plotted against $1-r$. We use $1-r$ rather than $r$ to facilitate comparison with $\beta$, since under this setup larger values of both parameters correspond to greater sensitivity to tail risk. The range of $1-r$ is restricted to values below $0.7$ to ensure the existence of the PHT risk measure. Unlike $d_\beta^*$, the optimal attachment point $d_{1-r}^*$ does not appear to be monotonic in $1-r$. The structural reasons for this behavior remain unclear and is worth further investigation.

Irrespective of its nonmonotonic behavior, Figure~\ref{fig:comparative_statics_tvar_pht} shows that the PHT-based optimal attachment point consistently exceeds its TVaR counterpart across the entire range of risk-sensitivity parameters considered. Intuitively, this observation is consistent with the stronger tail emphasis induced by the PHT distortion, which places substantially greater weight on extreme losses than TVaR. However, at present we lack a rigorous theoretical explanation for this ordering, and establishing such a result remains an important direction for future research.

More broadly, these findings motivate comparative statics with respect to structural model components rather than solely within a given parametric family. The comparative statics developed in the preceding sections examine how the optimal attachment point responds to changes in explicit model parameters, such as $\tau_2$, $w$, $\beta$, and $r$. A natural extension is to investigate how the optimal solution changes when the underlying risk evaluation framework itself is modified, for example, when moving between different families of risk measures (such as TVaR and PHT) or between alternative disutility functions. 
Another interesting direction of extension is to study the sensitivity of the optimal attachment point to the underlying loss distribution. This is of considerable practical importance because the loss distribution is typically  estimated from data and  therefore subject to model uncertainty and misspecification. Comparative statics with respect to distributional characteristics could provide valuable insights into the robustness of the optimal strategy and help practitioners assess the consequences of model error. Such results would contribute to a better understanding of how optimal CAT bond design should be implemented in the presence of distributional uncertainty.

\section{Conclusion}
\label{sec:conclussions}

In this paper, we investigate the optimal design of CAT bonds by treating the indemnity function as a decision variable rather than an exogenously specified feature. We analyze the design problem under both the expected disutility framework and the spectral risk measure framework. Our results show that the optimal indemnity function should follow the piecewise-linear structure under both frameworks, which reveals a common structural property that is robust across different risk evaluation criteria. In addition, we show that over single-layer indemnities with an unrestricted slope, maximizing the sponsor's expected net loss drives the design to a threshold contract (see Appendix~\ref{app:max_expected_net_loss}), which the slope bound in our admissible class rules out.

These findings have important implications for both practice and theory. From a practical perspective, The results provide a theoretical foundation for the sliding-scale payout mechanisms that have recently emerged in CAT bond markets and demonstrate their superior risk-reduction efficiency relative to traditional threshold payout structures. 
From a methodological perspective, the piecewise-linear characterization reduces the infinite-dimensional optimization problem to a one-dimensional optimization problem, which significantly enhances tractability and facilitates further analytical and numerical investigations of optimal CAT bond design.

Within the reduced one-dimensional optimization framework, we derive explicit or semi-explicit expressions for the optimal attachment point under each of the six risk criteria considered. All results admit a clear economic interpretation based on the trade-off between the marginal benefit and the marginal cost of increasing coverage. Building upon these results, we establish the monotonic behavior of the optimal attachment point with respect to changes in external parameters under the spectral framework, and numerically confirm it there while exploring it under the expected disutility framework.

Finally, we apply the model to U.S. severe storm loss data. The empirical results confirm our theoretical analysis and demonstrate the risk-reduction benefits of the optimal CAT bond design under each risk criterion. The case study also uncovers several empirical patterns that warrant further investigation. In particular, the optimal attachment point increases across several risk evaluation criteria, suggesting a tendency to shift coverage toward increasingly severe loss events as the risk criterion becomes more conservative. One possible explanation is that more punitive risk assessments induce the optimal contract to forgo protection against relatively small losses in order to preserve coverage for catastrophic, potentially solvency-threatening events. While this interpretation is economically intuitive, it remains an open question whether such a pattern holds more generally, and, if so, how the notion of  "more punitive"  should be formally defined and characterized within a theoretical framework.

The analysis developed in this paper lays a foundation for the study of optimal CAT bond design, but it also raises a number of challenging and practically relevant research questions. A first limitation of our framework is that both the issuance volume and the CAT bond price are treated as exogenous inputs. In practice, however, sponsors cannot freely choose an arbitrary level of coverage. The amount of protection that can be obtained depends on investor demand, while the attractiveness of alternative contract structures is ultimately determined by their market price. A natural next step is therefore to endogenize these quantities and study the joint determination of the indemnity function, issuance volume, and pricing mechanism. Such an extension would transform the design problem from a purely risk-management exercise into a market-equilibrium problem, providing a richer understanding of the interaction between sponsors' hedging objectives and investors' risk-return preferences.

A second direction concerns the choice of indemnity trigger. Throughout the paper we focus on the direct-loss indemnity trigger, which provides the cleanest benchmark for studying the optimal contract structure. In practice, however, many CAT bonds rely on parametric triggers, owing to their transparency, lower information requirements, and faster settlement process. These advantages come at the cost of basis risk, which introduces an additional trade-off absent from the current framework. It remains an open question how the optimal contract changes when payouts are linked to a parametric index rather than realized losses. More generally, future work could explore hybrid structures that combine indemnity-based and parametric triggers, and investigate the conditions under which the benefits of reduced basis risk justify the additional complexity.

A third direction concerns the extension of comparative statics analysis, which is essential for assessing the robustness of CAT bond designs. In this paper, we focus primarily on explicit model parameters, such as the issuance volume and the CAT bond pricing loading. A more challenging and arguably more interesting question is how the optimal attachment point responds to changes in the structural components of the model, including the risk evaluation criterion, represented by a disutility function or a spectral risk measure, and the distribution of the underlying loss. Such questions are of considerable practical importance. Sponsors often evaluate risk using multiple criteria to reflect different business purposes and the loss distribution itself is inevitably subject to estimation error and model uncertainty. Understanding how the optimal contract varies across alternative risk criteria and distributional assumptions would therefore provide valuable guidance for the design of robust CAT bond structures.

These directions point toward a broader research agenda in which contract design, risk measurement, model uncertainty, and market equilibrium are treated within a unified framework. We hope that the results and methodologies developed in this paper provide a useful first step toward that goal.



\bibliography{references_test}

\newpage

\appendix

\section{Expressions of $\operatorname{VaR}_\beta(L_d)$ and $\operatorname{TVaR}_\beta(L_d)$}
\label{app:netloss}

\begin{lemma}\label{lemma:netloss-dist}
The expression of $\operatorname{VaR}_\beta(L_d)$ is given by
\begin{eqnarray*}
\operatorname{VaR}_\beta(L_d) =
\begin{cases}
\operatorname{VaR}_\beta(X) - d + m(d), & 0\le\beta<F(d),\\
m(d), & F(d)\le\beta<F(d+w),\\
\operatorname{VaR}_\beta(X) - d - w + m(d), & F(d+w)\le\beta\le 1.
\end{cases}
\end{eqnarray*}
\end{lemma}

\begin{proof}
Recall from \eqref{piecewiseld} that 
\begin{equation*}
    L_d  = X-d - (X-d)_+\wedge w + m(d)
    =
    \left\{
     \begin{array}{ll}
        X - d + m(d), & \quad X\le d \\
        m(d), & \quad  d< X \le d+w \\
          X - d - w + m(d), & \quad X > d+w
     \end{array}
     \right. .
    \end{equation*}
Since $L_d = L_d(X)$ is an increasing function of $X$, according to the equivariance property of $\operatorname{VaR}$ under monotone transformation \citep{McNeil2015}, it follows that  $\operatorname{VaR}_\beta(L_d) = L_d(\operatorname{VaR}_\beta(X))$, which immediately yields the desired result.  

\end{proof}

\begin{lemma}\label{lemma:tvar_Ld}
The expression of $\operatorname{TVaR}_\beta(L_d)$ is given by
\begin{equation*}
\operatorname{TVaR}_\beta(L_d)
=
\begin{cases}
m(d) - d + \operatorname{TVaR}_\beta(X)
- \dfrac{\mathbb{E}\!\left[(X-d)_+\wedge w\right]}{1-\beta},
& \beta \le F(d),\\[10pt]
m(d) + \dfrac{\mathbb{E}\!\left[(X-(d+w))_+\right]}{1-\beta},
& F(d)< \beta\le F(d+w),\\[10pt]
m(d) - d + \operatorname{TVaR}_\beta(X) - w,
&\beta > F(d+w).
\end{cases}
\end{equation*}
\end{lemma}

\begin{proof}

Recall from \eqref{piecewiseld} that $L_d = X - (X-d)_+\wedge w - d + m(d)$. Note that $-d+m(d)$ is a deterministic quantity and both $X - (X-d)_+\wedge w$ and $(X-d)_+\wedge w$ are increasing in $X$. According to 
translation invariance  and
comonotonic additivity to of $\operatorname{TVaR}$ \citep{Acerbi2002}, we have 
\begin{equation*}\label{eq:split}
\operatorname{TVaR}_\beta(L_d)
= \operatorname{TVaR}_\beta(X) - \operatorname{TVaR}_\beta\big((X-d)_+\wedge w\big) - d + m(d).
\end{equation*}

Denote $Y=(X-d)_+\wedge w$. Recall that $\operatorname{TVaR}_\beta$ has risk spectrum function $\phi(u) = \frac{\mathbb{I}\{u>\beta\}}{1-\beta}$. Define $\Phi(u)=\int_0^u\phi(p)\,dp = \tfrac{(u-\beta)_+}{1-\beta}$.  
According to \cite{Acerbi2002}, $\operatorname{TVaR}$ admits the following expression
\[
\operatorname{TVaR}_\beta(Y) = \int_{0}^{\infty}\big[1 - \Phi(F_Y(y))\big]\,dy .
\]
Noting that 
\begin{equation*}
    F_Y(y) = 
    \begin{cases}
         F(d+y), & 0 \le y <w \\
        1, & y \ge w
    \end{cases},
\end{equation*}
we have 
\begin{equation}\label{eq:layerintegral}
\operatorname{TVaR}_\beta(Y) =\int_0^{d}\left[1 - \frac{(F(d+y)-\beta)_+}{1-\beta}\right]dy =  \int_d^{d+w}\left[1 - \frac{(F(x)-\beta)_+}{1-\beta}\right]dx.
\end{equation}
\smallskip

The evaluation of \eqref{eq:layerintegral} depends on the relative position of $\beta$ with respect to $F(d)$ and $F(d+w)$. We therefore consider the following cases.

\emph{Case $\beta\le F(d)$. 
}
The integral in \eqref{eq:layerintegral} reduces to  
\[ \frac{1}{1-\beta}\int_d^{d+w}(1-F(x))\,dx
= \frac{1}{1-\beta}\,\mathbb{E}[(X-d)_+\wedge w].\]

\emph{Case $\beta\ge F(d+w)$. 
}
In this case, $F(x) \le \beta$ for all  $x\in [d, d+w]$. Therefore, the integrand in \eqref{eq:layerintegral} equals $1$ throughout the integration interval, and thus the integral is equal to $w$.


\emph{Case $F(d)\le\beta\le F(d+w)$. 
}
In this case, we partition the integration interval into
$[d,\operatorname{VaR}_{\beta}(X)]$ and
$[\operatorname{VaR}_{\beta}(X),d+w]$.
Since the integrand in \eqref{eq:layerintegral} is constantly equal to $1$ on the first interval, the corresponding integral evaluates to
$\operatorname{VaR}_{\beta}(X)-d$.
The integral over the second interval is evaluated as follows:
\[
\begin{aligned}
\frac{1}{1-\beta}\int_{\operatorname{VaR}_\beta(X)}^{d+w}(1-F(x))\,dx
&= \frac{1}{1-\beta}\mathbb{E}\!\left[(X-\operatorname{VaR}_\beta(X))_+\right]
- \frac{1}{1-\beta}\mathbb{E}\!\left[(X-(d+w))_+\right],
\end{aligned}
\]
Combining the above two expressions and using the identity
$\operatorname{TVaR}_\beta(X)
=
\operatorname{VaR}_\beta(X)
+\frac{1}{1-\beta}\mathbb{E}\!\left[(X-\operatorname{VaR}_\beta(X))_+\right]$,
we obtain the desired expression for $\operatorname{TVaR}_\beta(L_d)$ in the case of $F(d)\le\beta\le F(d+w)$.
\end{proof}

\section{Limitation of the threshold payout structure}
\label{app:max_expected_net_loss}

In the introduction, we pointed out that the threshold payout structure may create moral hazard concerns. In this appendix, we take a closer look at this design and demonstrate an additional limitation. 
For this purpose, we first define the following class of indemnity functions:
\begin{equation}\nonumber
  \label{eq:admissible_class_C1}
  \mathfrak{C}_1
  :=
  \left\{
    I_{c,d} : [0,\infty) \to [0,1] \;\middle|\;
    I_{c,d}(x) := \frac{(x-d)_+}{c} \wedge 1,\;
    c>0,\; d \ge 0
  \right\}
\end{equation}

Clearly, when setting $c=w$, the indemnity function $I_{c,d}$ reduces to $I_d$ as defined in \eqref{optI}, which is the optimal form of the indemnity function for Problems \eqref{P-EDU} and \eqref{P-SPEC}. On the other hand, note that  
\begin{equation*}
    \lim_{c\to 0^+} I_{c,d}(x) = \mathbb{I}\{x>d\}.
\end{equation*}
In this sense, letting $c\to 0^+$, the indemnity function $I_{c,d}$ resembles the threshold payout structure.

Consider the net loss  $L_{c,d}$  under the indemnity function  $I_{c,d}$ and its expected value.  
\begin{eqnarray*}
     L_{c,d} &=& X- (1+\tau_1)\mathbb{E}[X] + (1+\tau_2)w\mathbb{E}[I_{c,d}(X)] - w I_{c,d}(X)\\
      J(c,d)
&:=& \mathbb{E}[L_{c,d}]
= -\tau_1\,\mathbb{E}[X]
  + w\tau_2\,\mathbb{E}\!\bigl[I_{c,d}(X)\bigr]
 \end{eqnarray*}

Note that the indemnity function $I_{c,d}(x)$ is decreasing in both parameters $c$ and $d$. Since $I_{c,d}(x)$ is bounded, according to the dominated convergence theorem, $\mathbb{E}[I_{c,d}(X)]$, and thus $\mathbb{E}[L_{c,d}]$, is decreasing in $c$ and $d$. Therefore, 
\[
\mathbb{E}[ \lim_{c\to 0^+} L_{c,d}(X)] \ge \mathbb{E}[  L_{w,d}(X)]  = \mathbb{E}[  L_{d}(X)].  
\]
In other words, any threshold indemnity function is inferior to the piecewise linear indemnity function defined in \eqref{optI} with the same attachment point. More strikingly, among all indemnity functions in $\mathfrak{C}_1$ with a common attachment point, the threshold structure is the least desirable choice since it maximizes the expected net loss. This fundamental inefficiency 
provides yet another reason to favor the piecewise linear indemnity function over the threshold structure in practice.

\end{document}